\documentclass[aps,prb,reprint,groupedaddress,onecolumn]{revtex4-2}

\usepackage{my}

\begin{document}

\title{Optimal statistical ensembles for quantum thermal state preparation \texorpdfstring{\\}{} within the quantum singular value transformation framework}
\author{Yasushi Yoneta}
\email{yasushi.yoneta@riken.jp}
\affiliation{RIKEN Center for Quantum Computing,
2-1 Hirosawa, Wako City, Saitama 351-0198, Japan}
\date{\today}
\begin{abstract}
Preparing thermal equilibrium states is an essential task
for finite-temperature quantum simulations.
In statistical mechanics, microstates in thermal equilibrium can be obtained from statistical ensembles.
To date, numerous ensembles have been devised,
ranging from Gibbs ensembles such as the canonical and microcanonical ensembles to a variety of generalized ensembles.
Since these ensembles yield equivalent thermodynamic predictions,
one can freely choose an ensemble for computational convenience.
In this paper, we exploit this flexibility
to develop an efficient quantum algorithm
for preparing thermal equilibrium states.
We first present a quantum algorithm
for implementing generalized ensembles
within the framework of quantum singular value transformation.
We then perform a detailed analysis of the computational cost
and elucidate its dependence on the choice of the ensemble.
Our analysis shows that employing an appropriate ensemble
can significantly mitigate ensemble-dependent overhead
and yield improved scaling of the computational cost with system size
compared to existing methods based on the canonical ensemble.
We also numerically demonstrate that
our approach achieves a significant reduction in the computational cost
even for small finite-size systems.
Our algorithm applies to arbitrary thermodynamic systems at any temperature
and is thus expected to offer a practical and versatile method
for computing finite-temperature properties of quantum many-body systems.
These results highlight the potential of ensemble design as a powerful tool
for enhancing the efficiency of a broad class of quantum algorithms.
\end{abstract}

\maketitle

\section{Introduction}
Quantum computers are expected to outperform classical computers in a wide class of computational tasks.
Among their applications, simulations of quantum many-body systems
stand out as one of the most promising~\cite{Feynman1982,Lloyd1996}.
In particular, finite-temperature simulations require preparation of thermal equilibrium states as initial states on quantum devices.
Therefore, developing efficient algorithms for thermal state preparation is essential for realizing practical finite-temperature quantum simulations.
However, computational complexity theory suggests that
preparing thermal equilibrium states of few-body Hamiltonians
is generally hard even for quantum computers,
even when restricted to classical Hamiltonians~\cite{Sly2010,Sly2012,Galanis2016}.
In fact, existing non-heuristic thermal state preparation algorithms
applicable to arbitrary systems at arbitrary temperatures
require computational cost that scales exponentially with the system size~\cite{Poulin2009,Chiang2010,Chowdhury2017,vanApeldoorn2020,vanApeldoorn2019,Gilyen2019}.

In statistical mechanics, microstates corresponding to thermal equilibrium states can be obtained from statistical ensembles~\cite{Gibbs1902,Landau1980,Toda1983,Callen1985}.
Due to their importance, various ensembles have been devised,
including the class of Gibbs ensembles such as the canonical and microcanonical ensembles~\cite{Gibbs1902},
as well as generalized ensembles that are obtained as extensions of the Gibbs ensembles~\cite{Hetherington1987,Challa1988_PRL,Challa1988_PRA,Johal2003,Ray1991,Gerling1993,Tsallis1988,Beck2003,Cohen2004,Costeniuc2005,Costeniuc2006,Toral2006,Penrose1971,Ellis2000,Touchette2010,Yoneta2019,Yoneta2023}.
When chosen appropriately, different ensembles yield equivalent thermodynamic predictions~\cite{Ruelle1999,MartinLof1979,Georgii1995,Lima1972,Mueller2015,Brandao2015,Tasaki2018}.
This allows one to freely select an ensemble
based on the physical situation and computational convenience.
In classical simulations, this flexibility has been exploited
to develop efficient algorithms by choosing suitable ensembles~\cite{Kim2010,Schierz2016,Sugiura2012}.

In contrast, most studies on quantum algorithms for thermal state preparation have employed only the canonical ensemble~\cite{Temme2011,Yung2012,Chen2021,Wocjan2023,Rall2023_2,Chen2023,Poulin2009,Chiang2010,Chowdhury2017,vanApeldoorn2020,vanApeldoorn2019,Gilyen2019,Bilgin2010}.
Although a few quantum algorithms utilizing generalized ensembles have been proposed~\cite{Seki2022,Mizukami2023},
these approaches directly adopt the generalized ensembles originally designed for classical algorithms.
However, the compatibility between ensembles and computational architectures or algorithmic frameworks can differ significantly.
Therefore, to develop efficient quantum algorithms, it is crucial to reconsider the choice of the ensemble.
As such, the flexibility in the choice of the ensemble has not yet been fully exploited for the development of quantum algorithms.

In this work, we explore optimal ensembles from a broad class of generalized ensembles
and develop an efficient quantum algorithm for thermal state preparation.
We adopt quantum singular value transformation (QSVT)~\cite{Gilyen2019,Martyn2021},
a unified framework for designing efficient quantum algorithms,
and construct a thermal state preparation algorithm based on generalized ensembles.
Our approach allows the ensemble-dependent computational overhead to be made arbitrarily small through the use of a suitably chosen ensemble,
thereby improving the scaling of the computational cost in the thermodynamic limit.
As a result, our method achieves a substantial reduction in the computational cost
compared to conventional approaches based on the canonical ensemble.
By providing guaranteed accuracy and applicability to arbitrary thermodynamic systems at any temperature, the proposed algorithm thus offers an efficient and versatile approach for finite-temperature quantum simulation.
In contrast to previous approaches,
which implicitly inherit the limitations of the canonical ensemble
due to its incompatibility with the constraint of unitarity in quantum computation,
our method explicitly constructs and leverages generalized ensembles to circumvent this bottleneck.
This perspective highlights ensemble design
as a key structural element in the development of efficient quantum algorithms.

The remainder of this paper is organized as follows.
In Section~\ref{sec:preliminaries},
we introduce the setting and notation,
review the formulation of statistical mechanics based on generalized ensembles,
and provide a brief overview of QSVT
focusing on the elements essential for our analysis.
In Section~\ref{sec:algorithm},
we present a quantum algorithm for implementing generalized ensembles.
In Section~\ref{sec:cost_TDL},
we perform an asymptotic analysis of the computational cost of our algorithm in the thermodynamic limit
and compare the performance achieved through the use of an optimal ensemble with that of existing methods.
We also discuss the origins of the advantages offered by employing generalized ensembles.
In Section~\ref{sec:cost_finite},
we numerically demonstrate that the advantage of our algorithm persists even for small finite-size systems.
Finally, in Section~\ref{sec:conclusion}, we provide our conclusions.

\section{Preliminaries} \label{sec:preliminaries}

\subsection{Setup} \label{sec:setup}
We introduce the setting and notations in this paper.
Let us consider a quantum spin-$1/2$ system with $N$ sites.
The Hamiltonian of the system is denoted by $H_N$.
We assume that the Hamiltonian $H_N$ is extensive in $N$, i.e., $\|H_N\|=\Theta(N)$.
Furthermore, we assume that the system is consistent with thermodynamics in the sense that
\begin{align}
    \frac{1}{N} \log g_N(u)
\end{align}
converges to an $N$-independent concave function
(the thermodynamic entropy density, $\s(u)$)
in the thermodynamic limit $N\to\infty$,
where $g_N(u)$ denotes the density of many-body microstates at the energy density $u$~\footnote{%
The assumption of the concavity of the limiting function is made to simplify the discussion.
When the limiting function is nonconcave, several difficulties arise,
such as the thermodynamic entropy becoming ill-defined
and the canonical ensemble giving a statistical mixture of macroscopically distinct states~\cite{Gross2001,Yoneta2019}.
However, even in such cases, all the results of this paper remain valid,
provided that $\s$ in the entropy formula~\eqref{eq:entropy-formula_canonical} for the canonical ensemble,
as well as all instances of $\s$ appearing in the asymptotic analysis of the computational cost for the thermal state preparation algorithm applied to the canonical ensemble,
are replaced with the concave envelope of the limiting function.
}.

Let $\I$ be the $2$-dimensional identity operator, and let $X$, $Y$, and $Z$ be the Pauli operators.
The eigenstates of $Z$ with eigenvalues $+1$ and $-1$ are denoted by $\ket{0}$ and $\ket{1}$, respectively.

\subsection{Statistical Ensemble} \label{sec:ensemble}
In statistical mechanics, thermal properties of many-body systems are derived using statistical ensembles.
A variety of ensembles have been introduced,
including Gibbs ensembles such as the canonical and microcanonical ensembles~\cite{Gibbs1902},
as well as their generalizations~\cite{Hetherington1987,Challa1988_PRL,Challa1988_PRA,Johal2003,Ray1991,Gerling1993,Tsallis1988,Beck2003,Cohen2004,Costeniuc2005,Costeniuc2006,Toral2006,Penrose1971,Ellis2000,Touchette2010,Yoneta2019,Yoneta2023}.
It is known that, as long as ensembles are appropriately chosen,
they yield equivalent results in the thermodynamic limit~\cite{Ruelle1999,MartinLof1979,Georgii1995,Lima1972,Mueller2015,Brandao2015,Tasaki2018}.
This equivalence allows one to select an ensemble based on computational convenience.
In this work, we exploit this flexibility to develop an efficient quantum algorithm for thermal state preparation.
To this end, we here review the formulation of statistical mechanics based on generalized ensembles.

\subsubsection{Canonical ensemble}
Among the numerous ensembles, the canonical ensemble is one of the most widely used.
To facilitate a comparison with generalized ensembles,
we first describe the canonical ensemble.
The density matrix for the canonical ensemble is given by
\begin{align}
    \rho_N^\mathrm{can}(\beta)
    = \frac{e^{- \beta H_N}}{\Z_N^\mathrm{can}(\beta)},
    \label{eq:def_density-matrix_canonical}
\end{align}
where $\Z_N^\mathrm{can}(\beta) = \mathrm{Tr}[e^{-\beta H_N}]$ is the partition function.
The canonical ensemble is parameterized by $\beta$,
which corresponds to the inverse temperature $1/T$ of the equilibrium state described by $\rho_N^\mathrm{can}(\beta)$~\cite{Callen1985}.
Therefore, the expectation values of local observables in $\rho_N^\mathrm{can}(\beta)$
coincide with their thermal expectation values at the inverse temperature $\beta$.
In particular, the energy density at the inverse temperature $\beta$ is given by
\begin{align}
    u^\mathrm{can}(\beta)
    = \lim_{N\to\infty} \mathrm{Tr} \left[ \rho_N^\mathrm{can}(\beta) H_N / N \right],
\end{align}
in the thermodynamic limit.

Using the partition function, the Helmholtz free energy is expressed as
\begin{align}
    f(\beta) = \lim_{N\to\infty} - \frac{1}{N\beta} \log \Z_N^\mathrm{can}(\beta).
    \label{eq:def_free-energy_canonical}
\end{align}
Therefore, from thermodynamics, the partition function is related to the thermodynamic entropy density $s$ via
\begin{align}
    s(u^\mathrm{can}(\beta))
    = \lim_{N\to\infty} \frac{1}{N} \log \Z_N^\mathrm{can}(\beta)
    + \beta u^\mathrm{can}(\beta).
    \label{eq:entropy-formula_canonical}
\end{align}

\subsubsection{Generalized ensemble} \label{sec:generalized-ensemble}
Beyond the canonical ensemble, numerous ensembles have been proposed.
In this paper, we consider a quite broad class of generalized ensembles~\cite{Costeniuc2005,Costeniuc2006,Yoneta2019},
for which the density matrix and partition function are given by
\begin{align}
    \rho_N^\eta &= \frac{e^{- N \eta(H_N/N)}}{\Z_N^\eta}, \label{eq:def_density-matrix}\\
    \Z_N^\eta &= \mathrm{Tr} \left[ e^{- N \eta(H_N/N)} \right]. \label{eq:def_partition-function}
\end{align}
Here, $\eta$ is a function independent of $N$ that satisfies the condition that $\s(u) - \eta(u)$ is strongly concave.
As mentioned in Section~\ref{sec:setup},
we assume that the thermodynamic entropy density $\s$ is concave,
which is a requirement for consistency with thermodynamic principles.
Thus, this condition is always satisfied if $\eta$ is chosen to be strongly convex.
Under this condition, the generalized ensemble describes the equilibrium state specified by the energy density $u_\mathrm{max}^\eta$,
which is defined as the unique maximum point of $\s - \eta$, in the thermodynamic limit:
\begin{align}
    u_\mathrm{max}^\eta = \argmax_{u} [\s(u)-\eta(u)].
    \label{eq:u-eta_argmax-s-minus-eta}
\end{align}
That is,
\begin{align}
    u^\eta = u_\mathrm{max}^\eta,
\end{align}
where
\begin{align}
    u^\eta = \lim_{N\to\infty} \mathrm{Tr} \left[ \rho_N^\eta H_N / N \right].
\end{align}
Therefore, one can obtain the desired equilibrium state by choosing $\eta$
such that $u_\mathrm{max}^\eta$ coincides with the energy density of the target state.

In addition, genuine thermodynamic quantities
such as the temperature and thermodynamic entropy
can also be computed using the generalized ensemble.
In contrast to the canonical ensemble, where the inverse temperature $\beta$ is a parameter,
in the generalized ensemble, $\beta$ is derived as a statistical mechanical quantity.
Unlike in the microcanonical ensemble,
where the inverse temperature is obtained by differentiating the thermodynamic entropy,
it can be computed directly from the expectation value of the energy density as
\begin{align}
    \beta = \eta'(u^\eta).
    \label{eq:temperature-formula}
\end{align}
Moreover, the thermodynamic entropy density $\s$ can be computed as
\begin{align}
    \s(u^\eta) = \lim_{N\to\infty} \frac{1}{N} \log \Z_N^\eta + \eta(u^\eta).
    \label{eq:entropy-formula}
\end{align}
Therefore, one can obtain all thermodynamic quantities using generalized ensembles,
as they can be computed from the thermodynamic entropy~\cite{Callen1985}.

In the class of generalized ensembles defined by Eqs.~\eqref{eq:def_density-matrix} and \eqref{eq:def_partition-function},
various ensembles are included as special cases.
In particular, by choosing
\begin{align}
  \eta(u) = \beta u + \mathrm{const.}, \label{eq:eta_canonical}
\end{align}
the generalized ensemble reduces to the canonical ensemble defined by Eq.~\eqref{eq:def_density-matrix_canonical}.
In this case, the temperature formula~\eqref{eq:temperature-formula} is consistent with the fact that the canonical ensemble describes the equilibrium state with the inverse temperature $\beta$.
Furthermore, Eq.~\eqref{eq:entropy-formula} reproduces the entropy formula for the canonical ensemble, presented in Eq.~\eqref{eq:entropy-formula_canonical}.
Thus, the canonical ensemble is a special case of the generalized ensemble.
However, in general, $\eta$ is not limited to linear functions.
Among the infinitely many possible choices of $\eta$,
one can select most suitable one
depending on the computational algorithm employed
or the physical situation under consideration.

\begin{example}[Gaussian ensemble~\cite{Hetherington1987,Challa1988_PRL,Challa1988_PRA,Johal2003,Costeniuc2005}]
Consider the generalized ensemble associated with the specific choice of $\eta$,
\begin{align}
    \eta(u) = \frac{1}{2} \lambda (u - \mu)^2,
\end{align}
where $\lambda$ is a positive constant, and $\mu$ is an ensemble parameter.
This choice of $\eta$ defines the so-called Gaussian ensemble
introduced in the earlier work by Hetherington~\cite{Hetherington1987}.
From Eq.~\eqref{eq:temperature-formula}, the inverse temperature $\beta$ is given as
\begin{align}
    \beta = \lambda ( u^\eta - \mu ).
\end{align}
Solving this for $\mu$, we find that setting $\mu = u^\eta - \beta / \lambda$ yields the equilibrium state at the inverse temperature $\beta$.
Conversely, as long as this condition is satisfied,
the parameter $\lambda$ can be chosen arbitrarily,
and any choice of $\lambda$ yields the equivalent result in the thermodynamic limit.
However, as will be shown in Section~\ref{sec:cost_TDL},
an appropriate choice of $\lambda$ is crucial for reducing the computational cost in quantum algorithms for thermal state preparation.
\end{example}

In the following, we focus on the case where $\eta$ is a polynomial.
This is because our algorithm is based on quantum singular value transformation,
which implements polynomial transformations,
implying that any function $\eta$ would ultimately need to be approximated by a polynomial.
We simplify the algorithm by restricting our discussion to polynomials from the outset.

\subsection{Quantum Singular Value Transformation} \label{sec:QSVT}
The optimal ensemble for thermal state preparation depends on the algorithmic framework employed.
In this work, we employ the quantum singular value transformation (QSVT) framework~\cite{Gilyen2019,Martyn2021}.
QSVT is a quantum algorithm that applies polynomial transformations to the singular values of a given matrix.
It was recently demonstrated that QSVT provides a unified framework
for understanding many prominent quantum algorithms discovered so far,
including Hamiltonian simulation~\cite{Low2017_2,Low2019,Gilyen2019},
amplitude amplification and estimation~\cite{Gilyen2019,Rall2023_1},
the HHL algorithm for solving quantum linear systems~\cite{Gilyen2019,Martyn2021},
and quantum phase estimation~\cite{Martyn2021,Rall2021}.
As such, QSVT serves as a primitive algorithm for designing efficient quantum algorithms.
In this section, we briefly describe QSVT techniques that will be used in our work.

\subsubsection{Quantum signal processing}
Before describing QSVT, we review quantum signal processing (QSP).
QSP provides efficient implementation of polynomial transformations of scalars,
and, consequently, serves as a foundational technique for polynomial transformations of the singular values of matrices.

Given a single-qubit reflection unitary that encodes $x \in [-1, +1]$, defined as
\begin{align}
    R(x)
    = \begin{pmatrix}
        x & \sqrt{1 - x^2} \\
        \sqrt{1 - x^2} & -x
    \end{pmatrix},
\end{align}
we consider constructing a unitary operator that encodes a polynomial transformation of $x$ as a product of $R(x)$ and single-qubit rotations.
\begin{theorem}[Quantum signal processing using reflections~\cite{Low2016,Gilyen2019}] \label{theorem:QSP}
Let $P(x) \in \mathbb{C}[x]$ be a degree-$d$ polynomial such that
\begin{enumerate}
    \item $P$ has parity $(d \bmod 2)$, \label{cond:QSP_1}
    \item for all $x \in [-1, +1]$, $|P(x)| \leq 1$, \label{cond:QSP_2}
    \item for all $x \in (-\infty, -1] \cup [+1, +\infty)$, $|P(x)| \geq 1$, \label{cond:QSP_3}
    \item if $d$ is even, then for any $x \in \mathbb{R}$, $P(ix) P^*(ix) \geq 1$. \label{cond:QSP_4}
\end{enumerate}
Then there exists $\Phi = (\phi_1, \phi_2, \cdots, \phi_d) \in \mathbb{R}^d$ such that
\begin{align}
    \prod_{j=1}^d \left( e^{i \phi_j Z} R(x) \right)
    = \begin{pmatrix}
        P(x) & \cdot \\
        \cdot & \cdot
    \end{pmatrix}.
\end{align}
\end{theorem}
By focusing on the real part of the polynomial $P$,
we obtain an intuitive condition for the class of polynomials that can be constructed.
\begin{theorem}[Real quantum signal processing~\cite{Gilyen2019}] \label{theorem:QSP_real}
Let $P_R(x) \in \mathbb{R}[x]$ be a degree-$d$ polynomial such that
\begin{enumerate}
    \item $P_R$ has parity $(d \bmod 2)$, \label{cond:QSP_real_1}
    \item for all $x \in [-1, +1]$, $|P_R(x)| \leq 1$. \label{cond:QSP_real_2}
\end{enumerate}
Then there exists a $P \in \mathbb{C}[x]$ with real part $P_R$,
satisfying conditions~\ref{cond:QSP_1}-\ref{cond:QSP_4} of Theorem~\ref{theorem:QSP}.
\end{theorem}
That is, by completing an appropriate imaginary part,
any real polynomial of definite parity can be implemented by QSP.
Importantly, one can find the coefficients of polynomial $P$ and the QSP phases $\Phi$
in polynomial time with a numerically stable classical algorithm~\cite{Gilyen2019,Haah2019,Chao2020,Dong2021,Ying2022,Berntson2024,Alexis2024}.

\subsubsection{Block-encoding} \label{sec:block-encoding}
QSVT is an extension of QSP, generalizing its applicability from scalars to matrices.
In quantum computations, all matrices should be described by the unitary operators.
Block-encoding is a technique to represent general matrices, which are not necessarily unitary, with unitary operators.

Intuitively, a block-encoding embeds the target matrix into a submatrix of a larger unitary matrix:
\begin{align}
    A \longmapsto U
    = \begin{pmatrix}
        A & \cdot\\
        \cdot & \cdot
    \end{pmatrix}
\end{align}
This is analogous to the embedding used in QSP, where scalars,
whose absolute values are not necessarily $1$,
are encoded in a single-qubit reflection unitary $R(x)$.

We give the formal definition of a block-encoding with the presence of errors as follows.
\begin{definition}[Block-encoding of a square matrix~\cite{Gilyen2019}] \label{def:block-encoding}
Suppose that $A$ is an $s$-qubit operator, $\alpha, \epsilon \in \mathbb{R}_{>0}$ and $a \in \mathbb{Z}_{\geq 0}$.
We say that the $(s+a)$-qubit unitary $U$ is an $(\alpha, a, \epsilon)$-block-encoding of $A$, if
\begin{align}
    \left\|
        A - \alpha (\bra{0}^{\otimes a} \otimes \I[s]) U (\ket{0}^{\otimes a} \otimes \I[s])
    \right\|
    \leq \epsilon.
\end{align}
\end{definition}
Here $a$ is the number of ancilla qubits for block-encoding,
and $\alpha$ is the subnormalization factor.
Due to the unitarity of $U$, the spectral norm of a submatrix of $U$ is upper bounded by $1$.
Consequently, if the target matrix $A$ has a norm greater than $1$,
we necessarily subnormalize $A$ with $\alpha$ so that $\| A \| \leq \alpha + \epsilon$.

The algorithm developed in this study is implemented by repeatedly querying the block-encoding of the Hamiltonian $H_N$, denoted by $U_H$,
and the overall computational cost is almost proportional to the number of queries to $U_H$.
Therefore, in this work, we assume access to $U_H$
and measure the computational cost of our algorithm in terms of the number of queries to $U_H$.
This approach enables us to assess the computational cost
without delving into the microscopic details of the system,
allowing us to concentrate on the search for the optimal ensemble.

\begin{example}
Let $A$ be an $s$-qubit operator
and let $\alpha \geq \|A\|$.
Consider the $(s+1)$-qubit operator
\begin{align}
    U = \begin{pmatrix}
        A/\alpha & \sqrt{\I[s] - (A/\alpha)(A/\alpha)^\dagger}\\
        \sqrt{\I[s] - (A/\alpha)^\dagger(A/\alpha)} & - (A/\alpha)^\dagger
    \end{pmatrix}.
\end{align}
A straightforward calculation shows that this operator satisfies
\begin{align}
    U U^\dagger = U^\dagger U = \I[(s+1)], \qquad
    (\bra{0} \otimes \I[s]) U (\ket{0} \otimes \I[s]) = A/\alpha.
\end{align}
Therefore, $U$ is an $(\alpha, 1, 0)$-block-encoding of $A$.
\end{example}

\subsubsection{Quantum singular value transformation}
We now describe QSVT.
First, we provide the definition of singular value transformation of matrices.
\begin{definition}[Singular value transformation~\cite{Gilyen2019}] \label{def:singular-value-transformation}
Let $A$ be a matrix with the singular value decomposition
\begin{align}
    A = \sum_i \varsigma_i \ket{\tilde{\psi_i}}\bra{\psi_i},
\end{align}
and let $f$ be a function of definite parity.
The singular value transformation of $A$ for $f$ is defined as
\begin{align}
    P^{(SV)}(A)
    = \begin{cases}
        \displaystyle \sum_i f(\varsigma_i) \ket{\tilde{\psi_i}}\bra{\psi_i} & \text{if $f$ is odd},\\ \\
        \displaystyle \sum_i f(\varsigma_i) \ket{\psi_i}\bra{\psi_i} & \text{if $f$ is even}.
    \end{cases}
\end{align}
\end{definition}

Given a block-encoding $U$ of a general matrix $A$,
QSVT enables the singular value transformation of $A$ for a polynomial $P$ of definite parity.
\begin{theorem}[Quantum singular value transformation~\cite{Gilyen2019,Martyn2021}] \label{theorem:QSVT}
Let $\mathcal{H}$ be a finite-dimensional Hilbert space.
Let $U$ be a unitary operator on $\mathcal{H}$,
and let $\Pi$ and $\tilde{\Pi}$ be orthogonal projectors on $\mathcal{H}$.
Let $P \in \mathbb{C}[x]$ be a degree-$d$ polynomial satisfying \ref{cond:QSP_1}-\ref{cond:QSP_4} of Theorem~\ref{theorem:QSP},
and let $\Phi \in \mathbb{R}^d$ be the corresponding sequence of phases.
Then
\begin{align}
    P^{(SV)} (\tilde{\Pi} U \Pi)
    = \begin{cases}
        \tilde{\Pi} U_{\Phi} \Pi & \text{if $d$ is odd},\\
        \Pi U_{\Phi} \Pi & \text{if $d$ is even}.
    \end{cases}
\end{align}
Here,
\begin{align}
    U_{\Phi}
    = \begin{cases}
        \displaystyle e^{i \phi_1 (2 \tilde{\Pi} - \I_\mathcal{H})} U \prod_{j = 1}^{(d-1)/2} \left( e^{i \phi_{2j} (2 \Pi - \I_\mathcal{H})} U^\dagger e^{i \phi_{2j+1} (2 \tilde{\Pi} - \I_\mathcal{H})} U \right) & \text{if $d$ is odd},\\
        \displaystyle \prod_{j = 1}^{d/2} \left( e^{i \phi_{2j-1} (2 \Pi - \I_\mathcal{H})} U^\dagger e^{i \phi_{2j} (2 \tilde{\Pi} - \I_\mathcal{H})} U \right) & \text{if $d$ is even}.
    \end{cases}
\end{align}
\end{theorem}
Therefore, when $U$ is a block-encoding of $A$ as $A / \alpha = \tilde{\Pi} U \Pi$ with $\tilde{\Pi} = \Pi = {\ket{0}\bra{0}}^{\otimes a} \otimes \I[s]$, the unitary $U_{\Phi}$ in this theorem gives a block-encoding of $P^{(SV)}(A / \alpha)$, i.e.,
\begin{align}
    P^{(SV)}(A / \alpha) = (\bra{0}^{\otimes a} \otimes \I[s]) U_{\Phi} (\ket{0}^{\otimes a} \otimes \I[s]).
\end{align}
It is worth noting that the operator $e^{i \phi (2 \Pi - \I_\mathcal{H})}$
(resp. $e^{i \phi (2 \tilde{\Pi} - \I_\mathcal{H})}$)
can be implemented efficiently using a single ancilla qubit,
two $\cnot{\Pi}$ (resp. $\cnot{\tilde{\Pi}}$),
and a single-qubit rotation~\cite{Gilyen2019},
using the relation
\begin{align}
    \sum_{b=0,1} \ket{b}\bra{b} \otimes e^{(-1)^b i \phi (2 \Pi - \I_\mathcal{H})}
    = \cnot{\Pi} (e^{-i \phi Z} \otimes \I_\mathcal{H}) \cnot{\Pi}.
    \label{eq:exp-i-phi-Pi}
\end{align}
Here, $\cnot{\Pi}$ is the $\Pi$-controlled NOT gate,
which is defined for an orthogonal projection operator $\Pi$ on $\mathcal{H}$ as
\begin{align}
    \cnot{\Pi} = X \otimes \Pi + \I \otimes (\I_\mathcal{H} - \Pi).
\end{align}

When $A$ is Hermitian, singular value transformation reduces to eigenvalue transformation.
In this case, the symmetry of $A$ allows us to remove the constraints on the parity of the polynomial.
This result can be summarized,
including an evaluation of the robustness of the transformation, as follows:
\begin{theorem}[Polynomial eigenvalue transformation of arbitrary parity~\cite{Gilyen2019}] \label{theorem:EVT}
Suppose that $U$ is an $(\alpha, a, \epsilon)$-block-encoding of a Hermitian matrix $A$.
If $P \in \mathbb{R}[x]$ is a degree-$d$ polynomial satisfying that
\begin{align}
    |P(x)| \leq 1/2 \label{eq:EVT_condition}
\end{align}
for all $x \in [-1,+1]$.
Then there exists a quantum circuit $\tilde{U}$,
which is an $(1,a+2,4d\sqrt{\epsilon/\alpha})$-encoding of $P(A/\alpha)$,
and which consists of $d-1$ applications of $U$ and $U^\dagger$,
a single application of either controlled-$U$ or controlled-$U^\dagger$,
and $\Order((a + 1)d)$ other one- and two-qubit gates.
\end{theorem}
For completeness, we provide the construction of the quantum circuit $\tilde{U}$ in Appendix~\ref{sec:EVT}.

Many existing algorithms naturally arise from QSVT
when applied with an appropriately chosen polynomial.
A prominent example of this is fixed-point amplitude amplification,
which plays a crucial role in preparing finite-temperature states~\cite{Poulin2009,Chiang2010,Chowdhury2017,vanApeldoorn2020,vanApeldoorn2019,Gilyen2019,Mizukami2023,Garratt2024}.
By taking the polynomial approximation of the sign function as in Ref.~\cite{Gilyen2019},
we prove the following theorem in Appendix~\ref{sec:proof_FPAA}:
\begin{theorem}[fixed-point amplitude amplification~\cite{Gilyen2019}] \label{theorem:FPAA}
Let $U$ be a unitary and $\Pi$ be an orthogonal projector
such that $\alpha \ket{\psi_G} = \Pi U \ket{\psi_0}$,
and $\alpha \geq \delta > 0$.
There exists a quantum circuit $\tilde{U}$
such that $\|\ket{\psi_G}-\tilde{U}\ket{\psi_0}\| \leq \epsilon$,
which uses a single ancilla qubit
and consists of $d$ applications of $U$ and $U^\dagger$,
$d+1$ applications of $\cnot{\Pi}$,
$d-1$ applications of $\cnot{\ket{\psi_0}\bra{\psi_0}}$,
and $\Order(d)$ other one- and two-qubit gates.
Here,
\begin{align}
    d
    = 2 \left\lceil \sqrt{t W \left( \frac{2^{16} k^2}{\pi t \epsilon^8} \right)} \right\rceil + 1
    = \Order\left( \frac{1}{\delta} \log \frac{1}{\epsilon} \right),
\end{align}
where
\begin{align}
    k = \frac{1}{\delta} \sqrt{\frac{1}{2} W\left( \frac{2^{11}}{\pi \epsilon^8} \right)}, \quad
    t = \left\lceil \max \left\{
        e^2 k^2 / 2, \log \frac{2^8 k}{\sqrt{\pi} \epsilon^4}
    \right\} \right\rceil.
\end{align}
Here, $W(x)$ is the Lambert-$W$ function, which satisfies $W(x) e^{W(x)} = x$.
\end{theorem}
We explain the implications of this theorem.
Suppose that a projective measurement described by $\Pi$ on the output state of the quantum circuit $U$
yields the post-measurement state $\ket{\Psi_G}$ with probability $p = |\alpha|^2$.
If we naively repeat this process until $\Pi$ is successfully measured,
the expected number of queries to $U$ before success is $\Order(p^{-1})$.
In contrast, by employing the above algorithm,
we can construct a quantum circuit that directly outputs the desired state $\ket{\Psi_G}$
with only $\Order(p^{-1/2})$ queries to $U$.
Notably, the query complexity of $\Order \left( \frac{1}{\alpha} \log \frac{1}{\epsilon} \right)$ achieved by this algorithm
is optimal with respect to both the target state overlap $\alpha$ and the error $\epsilon$~\cite{Yoder2014}.

As will be shown later, the computational cost of amplitude amplification dominates that of the thermal state preparation algorithm considered in this paper.
Therefore, rather than providing a mere order-of-magnitude estimate,
we present a precise evaluation of the computational cost.
A detailed proof of the theorem is given in Appendix~\ref{sec:proof_FPAA},
which also serves as a reference for the construction of the quantum circuit $\tilde{U}$.

\subsubsection{Polynomial approximation of the exponential function}
Since the computational cost of QSVT is determined by the degree of the polynomial,
it is crucial to obtain a low-degree polynomial
that approximates the desired function for the transformation.
In particular, for the application in this paper,
the polynomial approximation of the exponential function plays a key role.
As shown in Ref.~\cite{Sachdeva2014},
the exponential function can be approximated by polynomials
whose degrees are essentially quadratically smaller.
Here, we review an expansion in the basis of Chebyshev polynomials, which provides a concise representation of the explicit form.

\begin{lemma}[Polynomial approximation of the exponential function $e^{- \lambda (x+1)}$~\cite{Low2017_1}] \label{lemma:exp-function}
For all $\lambda > 0$, the polynomial $p_{\exp, \lambda, n}(x)$ of degree $n$
\begin{align}
    p_{\exp, \lambda, n}(x)
    = e^{- \lambda} \left(
        I_0(\lambda)
        + 2 \sum_{j=1}^{n} I_j(\lambda) T_j(-x)
    \right)
    \label{eq:p_exp}
\end{align}
satisfies
\begin{align}
    \epsilon_{\exp, \lambda, n}
    &= \max_{x \in [-1, +1]}
    \left|
        p_{\exp, \lambda, n}(x)
        - e^{- \lambda (x+1)}
    \right|
    \leq 2 e^{- n^2 / 2 t} + e^{- \lambda - t}
\end{align}
for every integer $t \geq e^2 \lambda$.
Here, $I_j(x)$ are modified Bessel functions of the first kind,
and $T_j(x)$ are Chebyshev polynomials of the first kind.
\end{lemma}
Therefore, if we set
\begin{align}
    t = \left\lceil \max \left\{
        e^2 \lambda,\log \frac{2}{\epsilon}
    \right\} \right\rceil, \quad
    n = \left\lceil \sqrt{2 t \log \frac{4}{\epsilon}} \right\rceil,
\end{align}
then it holds that $\epsilon_{\exp, \lambda, n} \leq \epsilon/2 + \epsilon/2 = \epsilon$.

\section{Thermal state preparation based on QSVT} \label{sec:algorithm}
We now present an algorithm within the QSVT framework for preparing a thermal equilibrium state on a quantum computer using a generalized ensemble.
The main result of this section is summarized in the following theorem:
\begin{theorem}[Thermal state preparation with the generalized ensemble] \label{theorem:main}
Suppose that the Hamiltonian $H_N$ is a Hermitian operator.
Let $\epsilon > 0$, $\eta \in \mathbb{R}$ be a degree-$d_\eta$ real polynomial.
Let $U_H$ be a block-encoding of $H_N$ using $a$ ancilla qubits,
\begin{align}
    H_N = \alpha N (\bra{0}^{\otimes a} \otimes \I[N]) U_H (\ket{0}^{\otimes a} \otimes \I[N]).
\end{align}
Then, we can implement a unitary circuit $U$ such that
\begin{align}
    \left\|
        \ket{0}^{\otimes(a+3)} \otimes \ket{\rho_N^\eta}
        - U \ket{0}^{\otimes (2N+a+3)}
    \right\|
    \leq \epsilon,
\end{align}
where $\ket{\rho_N^\eta}$ is a purification of the generalized ensemble $\rho_N^\eta \propto e^{- N \eta(H_N / N)}$ associated with $\eta$,
with $d_\eta d_\mathrm{exp} d_\mathrm{AA}$ queries to (controlled) $U_H$ and ${U_H}^\dagger$.
Here,
\begin{align}
    d_\mathrm{exp}
    = \left\lceil
        \sqrt{2 \left\lceil
            \max \left\{
                \frac{1}{4} e^2 N (\eta_\mathrm{max}-\eta_\mathrm{min}),
                \log\frac{2^3}{\epsilon \sqrt{\zeta}}
            \right\}
        \right\rceil
        \log\frac{2^4}{\epsilon \sqrt{\zeta}}}
    \right\rceil, \qquad
    d_\mathrm{AA}
    = 2 \left\lceil \sqrt{t W \left( \frac{2^{24} k^2}{\pi t \epsilon^8} \right)} \right\rceil + 1,
\end{align}
where
\begin{align}
    \eta_\mathrm{max} &= \max_{x\in[-1,+1]} \eta(\alpha x), \qquad
    \eta_\mathrm{min} = \min_{x\in[-1,+1]} \eta(\alpha x),\\
    \zeta = \frac{e^{N \eta_\mathrm{min}} \Z_N^\eta}{2^N}, \qquad
    k &= \frac{1}{1 - \epsilon / 2} \sqrt{\frac{2}{\zeta} W\left( \frac{2^{19}}{\pi \epsilon^8} \right)}, \qquad
    t = \left\lceil \max \left\{ e^2 k^2 / 2, \log \frac{2^{12} k}{\sqrt{\pi} \epsilon^4} \right\} \right\rceil.
\end{align}
\end{theorem}

In the following, we explicitly present an algorithm
for constructing the generalized ensemble
and provide a proof of Theorem~\ref{theorem:main}.
Here, we consider a general case where the block-encoding of $H_N$ contains errors.
Specifically, let $U_H$ be an $(\alpha N, a, \epsilon_H)$-block-encoding of the Hamiltonian $H_N$.

\begin{figure*}[t]
    \begin{minipage}{0.54\linewidth}
        \centering
        \includegraphics[height=5.5cm]{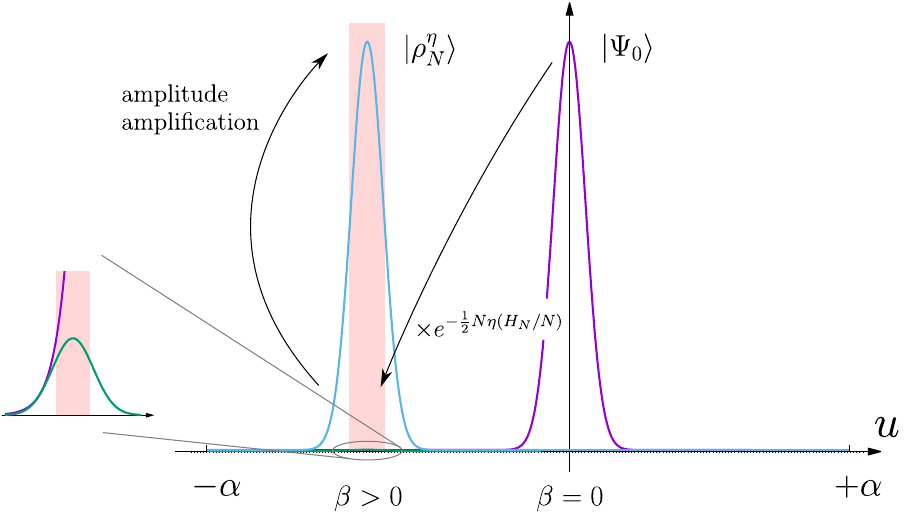}
        \subcaption{}
        \label{fig:distribution}
    \end{minipage}
    \hfill
    \begin{minipage}{0.44\linewidth}
        \centering
        \includegraphics[height=5.5cm]{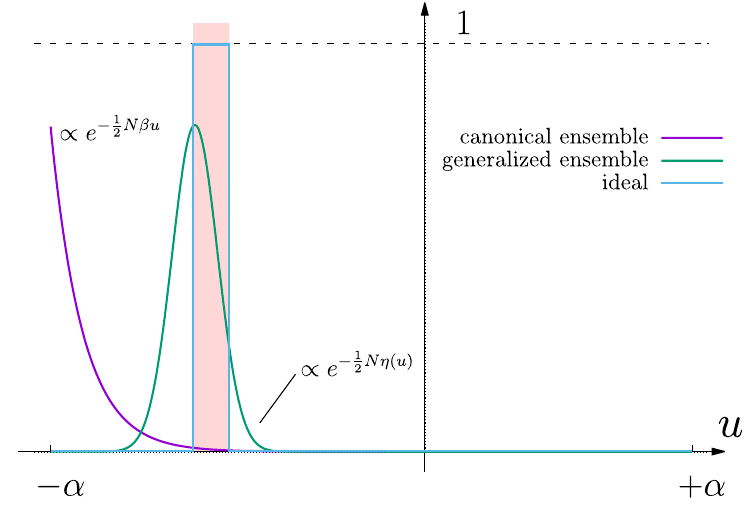}
        \subcaption{}
        \label{fig:filter}
    \end{minipage}
    \caption{%
        (\subref{fig:distribution})
        Transition of the energy density distribution in the thermal state preparation algorithm presented in Section~\ref{sec:algorithm}.
        First, in Step~\hyperref[sec:step1]{1}, the purification of the infinite-temperature state $\ket{\Psi_0}$ is prepared (purple curve).
        Next, in Step~\hyperref[sec:step2]{2}, the application of $e^{- \frac{1}{2} N \eta(H_N/N)}$ suppresses high-energy components (green curve),
        shifting the peak of the distribution toward the target energy region at finite temperature (red shaded area).
        Finally, in Step~\hyperref[sec:step3]{3}, amplitude amplification is performed so that the desired state can be obtained in a single shot (light blue curve).\\
        (\subref{fig:filter})
        Filter function applied to the expansion coefficients in the energy eigenbasis in Step~\hyperref[sec:step2]{2}.
        To be embedded into a unitary operator, the filter function must not exceed $1$ (gray dashed line).
        To reduce the computational cost of the subsequent amplitude amplification,
        the filter should be as large as possible, ideally close to $1$, in the target energy region (red shaded area).
        However, when the canonical ensemble is employed, the filter function inevitably becomes exponentially small, since it takes its maximum value at the ground state energy and decays exponentially from there (purple curve).
        In contrast, when generalized ensembles are employed, the filter function can be made significantly larger in the target energy region (green curve).
    }
\end{figure*}

\subsection*{Step~1: Preparation of a Maximally Entangled State} \label{sec:step1}
To coherently prepare the generalized ensemble, which is a mixed state,
we adopt the purification method~\footnote{%
Similarly, one can alternatively employ the pure-state method~\cite{Sugiura2012,Sugiura2013,Yoneta2024}.
While this approach requires additional computational cost for preparing infinite-temperature states,
it allows one to eliminate the ancilla register of $N$ qubits that is required for purification.
Even in this case, however, the number of queries to $U_H$ required to construct each individual state remains identical to that in the purification method.
}.
Then we introduce a quantum register $S$ consisting of $N$ qubits that describes the system and an ancilla register $A$ of the same size $N$.
Suppose that at the beginning of the computation the registers are initialized to
\begin{align}
    \left(\bigotimes_{n=1}^N \ket{0}_{n}\right) \otimes \left(\bigotimes_{n=1}^N \ket{0}_{\overline{n}}\right),
\end{align}
where $\ket{\sigma}_{n}$ and $\ket{\sigma}_{\overline{n}}$ ($\sigma=0, 1$) denotes the computational basis state of the $n$-th qubit of $S$ and $A$, respectively.

As the first step, we prepare a maximally entangled state between $S$ and $A$ given by
\begin{align}
  \ket{\Psi_0}
  = \frac{1}{2^{N/2}} \sum_{\substack{(\sigma_1,\sigma_2,\cdots,\sigma_N)\\\in {(0,1)}^N}} \left(\bigotimes_{n=1}^N \ket{\sigma_n}_{n}\right)
  \otimes \left(\bigotimes_{n=1}^N \ket{\sigma_n}_{\overline{n}}\right).
\end{align}
This state can be produced from the initial state via a single layer of CNOT gates since it can be expressed as
\begin{align}
  \ket{\Psi_0}
  = \bigotimes_{n=1}^N \frac{\ket{0}_{n}\ket{0}_{\overline{n}}+\ket{1}_{n}\ket{1}_{\overline{n}}}{\sqrt{2}}
  = \bigotimes_{n=1}^N \cnot{}_{n,\overline{n}} \ \hadamard_{n} \ket{0}_{n}\ket{0}_{\overline{n}}.
\end{align}
Here, $\cnot{}_{n, \overline{n}}$ denotes the CNOT gate with the $n$-th qubit of the system register being the control qubit and the $n$-th qubit of the ancilla register being the target qubit, and $\hadamard_{n}$ denotes the Hadamard gate acting on the $n$-th qubit of the system register.

It is worth noting that $\ket{\Psi_0}$ is a purification of the Gibbs state $\propto e^{- \beta H_N}$ at the inverse temperature $\beta=0$ (i.e., $T\to\infty$).
In the subsequent steps, we will ``cool'' $\ket{\Psi_0}$ and obtain a finite-temperature state.

\subsection*{Step~2: Construction of the Block-Encoding of \texorpdfstring{$\exp[-\frac{1}{2}N\eta(H_N/N)]$}{exp[- 1/2 N eta(H/N)]}} \label{sec:step2}
As the second step, we construct the block-encoding of $\exp[-\frac{1}{2}N\eta(H_N/N)]$ via the QSVT technique.

We introduce a degree-$d_\eta$ polynomial $\tilde{\eta}$ by normalizing $\eta$ such that $|\tilde{\eta}(x)| \leq 1$ for $x \in [-1,+1]$:
\begin{align}
    \tilde{\eta}(x) = \frac{2\eta(\alpha x)-(\eta_\mathrm{max}+\eta_\mathrm{min})}{\eta_\mathrm{max}-\eta_\mathrm{min}}.
\end{align}
Here,
\begin{align}
    \eta_\mathrm{max} = \max_{x\in[-1,+1]} \eta(\alpha x), \qquad
    \eta_\mathrm{min} = \min_{x\in[-1,+1]} \eta(\alpha x).
\end{align}
According to Lemma~\ref{lemma:exp-function},
for $\lambda = \frac{1}{4} N (\eta_\mathrm{max}-\eta_\mathrm{min})$
and any $\epsilon_{\mathrm{exp}} > 0$,
there exists a degree-$d_{\mathrm{exp}}$ polynomial $p_{\exp, \lambda, d_{\mathrm{exp}}}$ that satisfies
\begin{align}
    \max_{x \in [-1, +1]} \left| p_{\exp, \lambda, d_{\mathrm{exp}}}(x) - e^{- \lambda (x+1)} \right|
    \leq \epsilon_{\mathrm{exp}},
\end{align}
where
\begin{align}
    d_\mathrm{exp}
    = \left\lceil
        \sqrt{2 \left\lceil
            \max \left\{
                \frac{1}{4} e^2 N (\eta_\mathrm{max}-\eta_\mathrm{min}),
                \log\frac{2}{\epsilon_\mathrm{exp}}
            \right\}
        \right\rceil
        \log\frac{4}{\epsilon_\mathrm{exp}}}
    \right\rceil.
\end{align}
As shown in Appendix~\ref{sec:proof_p-exponential_bound}, it holds that
\begin{align}
    \max_{x \in [-1, +1]} |p_{\exp, \lambda, d_{\mathrm{exp}}}(x)| \leq 1.
    \label{eq:p-exponential_bound}
\end{align}
We then define a degree-$d_\eta d_\mathrm{exp}$ polynomial $P$ as
\begin{align}
    P(x) = \frac{1}{2} p_{\exp, \lambda, d_{\mathrm{exp}}}(\tilde{\eta}(x)),
\end{align}
which satisfies
\begin{align}
    \max_{x \in [-1, +1]} |P(x)| \leq 1/2, \qquad
    \max_{x \in [-1, +1]} \left| P(x) - \frac{1}{2} e^{- \frac{1}{2} N [\eta(\alpha x)-\eta_\mathrm{min}]} \right|
    \leq \epsilon_\mathrm{exp} / 2.
\end{align}
Therefore, by Theorem~\ref{theorem:EVT}, we can compose a quantum circuit $V$
that implements an $(1, a+2, 4 d_\mathrm{exp} d_\eta \sqrt{\epsilon_H/\alpha N} + \epsilon_\mathrm{exp} / 2)$-block-encoding of $\frac{1}{2} e^{- \frac{1}{2} N [\eta(H_N/N)-\eta_\mathrm{min}]}$,
which consists of $d_\eta d_\mathrm{exp}-1$ applications of $U_H$ and ${U_H}^\dagger$ and a single application of controlled $U_H$.

\subsection*{Step~3: Amplitude Amplification} \label{sec:step3}
As the third step, we apply $e^{- \frac{1}{2} N \eta(H_N/N)}$ to $\ket{\Psi_0}$, thereby obtaining the purification of $\rho_N^\eta$.

For the quantum circuit $V$ obtained in the previous step, we have
\begin{align}
    \left\|
        \frac{1}{2} \sqrt{\zeta}
        \ket{0}^{\otimes (a+2)} \otimes \ket{\rho_N^\eta}
        - \left({\ket{0}\bra{0}}^{\otimes (a+2)} \otimes \I[2N]\right)
        \left(V \otimes \I[N]\right)
        \left(\ket{0}^{\otimes (a+2)} \otimes \ket{\Psi_0}\right)
    \right\|
    \leq 4 d_\mathrm{exp} d_\eta \sqrt{\epsilon_H/\alpha N} + \epsilon_\mathrm{exp} / 2,
\end{align}
where $\zeta = e^{N \eta_\mathrm{min}} \Z_N^\eta / 2^N$,
and $\ket{\rho_N^\eta} = \sqrt{\rho_N^\eta} \otimes \I[N] \ket{\Psi_0}$ is a purification of $\rho_N^\eta$, which satisfies
\begin{align}
    \mathrm{Tr}_{A} [ \ket{\rho_N^\eta}\bra{\rho_N^\eta} ]
    = \rho_N^\eta.
\end{align}
In other words, when the quantum circuit $U = V \otimes \I[N]$ is applied
to $\ket{0}^{\otimes (a+2)} \otimes \ket{\Psi_0}$,
and the measurement operator $\Pi={\ket{0}\bra{0}}^{\otimes (a+2)} \otimes \I[2N]$ is applied to the output state, the resulting post-measurement state approximates the desired state $\ket{\rho_N^\eta}$.

However, its success probability $p$ is exponentially small in the system size $N$
\begin{align}
    p &= \left\|
        \left({\ket{0}\bra{0}}^{\otimes (a+2)} \otimes \I[2N]\right) \left(V \otimes \I[N]\right) \left(\ket{0}^{\otimes (a+2)} \otimes \ket{\Psi_0}\right)
    \right\|^2
    \simeq \zeta / 4
    = \exp[-\Theta(N)]
\end{align}
when
\begin{align}
  4 d_\mathrm{exp} d_\eta \sqrt{\epsilon_H/\alpha N} + \epsilon_\mathrm{exp} / 2 \ll \sqrt{\zeta} / 2.
\end{align}
This condition is equivalent to the error in the block-encoding $V$ being sufficiently smaller than the Frobenius norm of the embedded matrix $\frac{1}{2} e^{- \frac{1}{2} N [\eta(H_N/N)-\eta_\mathrm{min}]}$, which is necessarily satisfied when aiming for an accurate simulation.
Consequently, if one were to simply repeat the computation until the measurement succeeds, the expected number of repetitions required would be $1/p = \exp[\Theta(N)]$.

Fortunately, the amplitude amplification algorithm provides a quadratic speedup, significantly reducing the computational cost.
Noting that
\begin{align}
    \left\|
        \left({\ket{0}\bra{0}}^{\otimes (a+2)} \otimes \I[2N]\right) \left(V \otimes \I[N]\right) \left(\ket{0}^{\otimes (a+2)} \otimes \ket{\Psi_0}\right)
    \right\|
    \geq \sqrt{\zeta} /2 - \left( 4 d_\mathrm{exp} d_\eta \sqrt{\epsilon_H/\alpha N} + \epsilon_\mathrm{exp} / 2 \right),
\end{align}
and applying Theorem~\ref{theorem:FPAA}, we obtain the following result:
We can build a unitary quantum circuit $W$ such that
\begin{align}
    \left\|
        \ket{0} \otimes \frac{\left({\ket{0}\bra{0}}^{\otimes (a+2)} \otimes \I[2N]\right) \left(V \otimes \I[N]\right) \left(\ket{0}^{\otimes (a+2)} \otimes \ket{\Psi_0}\right)}{\left\|\left({\ket{0}\bra{0}}^{\otimes (a+2)} \otimes \I[2N]\right) \left(V \otimes \I[N]\right) \left(\ket{0}^{\otimes (a+2)} \otimes \ket{\Psi_0}\right)\right\|}
        - W \left(\ket{0}^{\otimes (a+3)} \otimes \ket{\Psi_0}\right)
    \right\|
    \leq \epsilon_\mathrm{AA},
\end{align}
which consists of $d_\mathrm{AA} = 2 \left\lceil \sqrt{t W \left( \frac{2^{16} k^2}{\pi t {\epsilon_\mathrm{AA}}^8} \right)} \right\rceil + 1$ applications of $V$ and $V^\dagger$.
Here,
\begin{align}
    k &= \frac{1}{\delta} \sqrt{\frac{1}{2} W\left( \frac{2^{11}}{\pi {\epsilon_\mathrm{AA}}^8} \right)}, \qquad
    t = \left\lceil \max \left\{ e^2 k^2 / 2, \log \frac{2^8 k}{\sqrt{\pi} {\epsilon_\mathrm{AA}}^4} \right\} \right\rceil, \qquad
    \delta = \sqrt{\zeta} / 2 - 2 \left( 4 d_\mathrm{exp} d_\eta \sqrt{\epsilon_H/\alpha N} + \epsilon_\mathrm{exp} / 2 \right).
\end{align}
Thus, with $d_\mathrm{AA}$ queries to $V$, one can prepare a state that approximate the purification of the generalized ensemble $\rho_N^\eta$:
\begin{align}
    &\left\|
    \ket{0}^{\otimes(a+3)} \otimes \ket{\rho_N^\eta}
    - W \left( \ket{0}^{\otimes(a+3)} \otimes \ket{\Psi_0} \right)
    \right\|
    \leq 2 \left( \frac{4 d_\mathrm{exp} d_\eta \sqrt{\epsilon_H/\alpha N} + \epsilon_\mathrm{exp} / 2}{\frac{1}{2} \sqrt{\zeta}} \right) + \epsilon_\mathrm{AA}.
\end{align}
By choosing
\begin{align}
    \epsilon_\mathrm{exp} = \epsilon \sqrt{\zeta} / 4,\qquad
    \epsilon_\mathrm{AA} = \epsilon / 2,
\end{align}
it is ensured that
\begin{align}
    \left\|
        \ket{0}^{\otimes(a+3)} \otimes \ket{\rho_N^\eta}
        - W \left( \ket{0}^{\otimes(a+3)} \otimes \ket{\Psi_0} \right)
    \right\|
    \leq \epsilon + 16 d_\mathrm{exp} d_\eta \sqrt{\epsilon_H/\alpha N \zeta}.
\end{align}

In this way, with $d_\eta d_\mathrm{exp} d_\mathrm{AA}$ queries to (controlled) $U_H$ and ${U_H}^\dagger$, one can prepare the purification of the generalized ensemble.
In particular, in the case of $\epsilon_H = 0$, we obtain Theorem~\ref{theorem:main}.

\section{Computational Cost in the Thermodynamic Limit} \label{sec:cost_TDL}
In this section, we integrate the results of the previous section with insights from statistical mechanics
to clarify how the computational cost of thermal state preparation depends on the choice of the ensemble.
Our analysis reveals that employing an appropriate ensemble leads to an improvement in the scaling of the computational cost compared to existing methods based on the canonical ensemble, as well as the generalized ensembles used in the previous study.
Furthermore, we discuss the origin of the advantage provided by utilizing appropriate generalized ensembles.

\subsection{Asymptotic Analysis} \label{sec:asymptotic-analysis}
We first perform an asymptotic analysis of the computational cost.
From the bound
\begin{align}
    \log x - \log \log x \leq W(x) \leq \log x - \frac{1}{2} \log \log x
\end{align}
for $x \geq e$~\cite{Hoorfar2008},
it follows that the number of queries to the block-encoding $U_H$ of the Hamiltonian required by the algorithm of Theorem~\ref{theorem:main} is
\begin{align}
    (\text{number of queries to $U_H$})
    &= \Order \left(
        d_\eta
        \sqrt{
            \frac{1}{\zeta}
            \left( N (\eta_\mathrm{max}-\eta_\mathrm{min}) + \log \frac{1}{\epsilon} + \log \frac{1}{\zeta} \right)
            \left( \log \frac{1}{\epsilon} + \log \frac{1}{\zeta} \right)
        }
        \log \frac{1}{\epsilon}
    \right). \label{eq:scaling}
\end{align}
This makes it clear that the computational cost for constructing the generalized ensemble is determined by three quantities:
the degree of the polynomial $\eta$, $d_\eta$;
the oscillation of $\eta$ over the interval $[-\alpha,+\alpha]$, $\eta_\mathrm{max} - \eta_\mathrm{min}$;
and the subnormalized partition function $\zeta = e^{N \eta_\mathrm{min}} \Z_N^\eta / 2^N$.
Since $\zeta$ depends on the thermodynamic properties of the system,
its ensemble dependence is not immediately apparent.
In the following, we incorporate results from statistical mechanics to elucidate the ensemble dependence of the computational cost.

From the entropy formula~\eqref{eq:entropy-formula}, $\zeta$ can be expressed in terms of the thermodynamic entropy density $\s$ and the energy density $u^\eta$ of the equilibrium state described by the generalized ensemble as
\begin{align}
    \zeta
    &= \frac{e^{N \eta_\mathrm{min}} \Z_N^\eta}{2^N}
    = \frac{e^{N \s(u^\eta)}}{2^N}
    \times e^{- N [\eta(u^\eta) - \eta_\mathrm{min}]}
    \times e^{o(N)}.
\end{align}
Thus, we obtain
\begin{align}
    (\text{number of queries to $U_H$})
    &= \underbrace{\sqrt{\frac{2^N}{e^{N \s(u^\eta)}}}}_{A_N^\eta}
    \times \underbrace{e^{\frac{1}{2} N [\eta(u^\eta) - \eta_\mathrm{min}]}}_{B_N^\eta}
    \times e^{o(N)}.
    \label{eq:scaling_N}
\end{align}
This shows that the leading factor of the computational cost grows exponentially with $N$.
We then discuss the ensemble dependence of the factors $A_N^\eta$ and $B_N^\eta$.

The first factor, $A_N^\eta$, depends only on the properties of the target equilibrium state.
This factor implies that equilibrium states with larger entropy are easier to prepare.
In particular, since the number of microstates with energy $u^\eta$ is given by $D \sim \exp[N\s(u^\eta)]$ according to the Boltzmann's entropy formula,
$A_N^\eta$ coincides with the complexity of searching for $D$ states within the target energy region among the total $2^N$ states in the whole Hilbert space using a quantum search algorithm~\cite{Grover1996,Grover2005,Hoyer2000,Long2001,Yoder2014}.
While classical search requires $\Order(2^N/\exp[N \s(u)])$ queries, quantum search enjoys a quadratic speedup achieved via amplitude amplification, reducing the cost to $\Order(\sqrt{2^N/\exp[N \s(u)]})$.
Moreover, it is important to note that classical methods are directly applicable only to diagonal Hamiltonians and, in the general case, require diagonalization of the Hamiltonian, which has an $\Order(2^N)$ space complexity and an even worse time complexity.

The second factor, $B_N^\eta$, is an artificial, exponentially large factor that depends on the choice of the ensemble (i.e., the function $\eta$).
This factor can be made small by choosing $\eta$ appropriately.
Thus, to optimize the scaling of the leading factor in the computational cost,
one should choose $\eta$ so as to minimize $\eta(u^\eta) - \eta_\mathrm{min}$, which appears in the exponent,
among those that yield the same equilibrium state.
As will be illustrated in Section~\ref{sec:optimal-ensemble}, this can be made arbitrarily small.

Therefore, by employing a suitably chosen ensemble, the total computational cost of our algorithm scales as
\begin{align}
    (\text{number of queries to $U_H$})
    = \sqrt{\frac{2^N}{e^{N \s(u^\eta)}}}
    \times e^{N \varepsilon + o(N)},
    \label{eq:optimal-scaling_N}
\end{align}
where $\varepsilon$ is an arbitrarily small positive constant independent of $N$.

\subsection{Canonical Ensemble}
For comparison, we examine the computational cost of preparing the canonical ensemble.
As explained in Section~\ref{sec:generalized-ensemble}, the generalized ensemble $\rho_N^\eta$ reduces to the canonical ensemble when $\eta = \beta u + \mathrm{const.}$.
Furthermore, the algorithm described in Section~\ref{sec:algorithm} reduces to the thermal state preparation algorithm based on the canonical ensemble proposed in Ref.~\cite{Gilyen2019}.
Hence, the computational cost is given by Theorem~\ref{theorem:main},
and its ensemble dependence is dominated by the factor
\begin{align}
    B_N^\eta = \exp \left[ \frac{1}{2} N \beta (u^\mathrm{can}(\beta) + 1) \right].
\end{align}
Since $\beta$ is the inverse temperature of the target equilibrium state and is uniquely determined for the equilibrium state, the exponent of $B_N^\eta$, $\frac{1}{2} N \beta [u^\mathrm{can}(\beta) + 1] (>0)$ is fixed, and there is no room to reduce it further.

\subsection{Optimal Ensemble} \label{sec:optimal-ensemble}
Using the results of Section~\ref{sec:asymptotic-analysis},
we search for ensembles that attain the optimal scaling of computational cost,
as given in Eq.~\eqref{eq:optimal-scaling_N},
in our thermal state preparation algorithm.

Recall that the ensemble dependence of the computational cost is dominated by
\begin{align}
    B_N^\eta = \exp \left[ \frac{1}{2} N \left( \eta(u^\eta) - \eta_\mathrm{min} \right) \right].
\end{align}
Thus, it suffices to search for the function $\eta$
that minimizes $\eta(u^\eta) - \eta_\mathrm{min}$
among those that describe the desired equilibrium state.

We begin by examining the qualitative behavior of $\eta(u^\eta) - \eta_\mathrm{min}$.
From Eq.~\eqref{eq:u-eta_argmax-s-minus-eta}, it follows that $u^\eta$ minimizes the function $\eta - \s$.
Consequently, if the contribution of the entropy term $\s$ is negligible, $u^\eta$ will be located near the minimum of $\eta$, leading to $\eta(u^\eta) - \eta_\mathrm{min} \simeq 0$.
However, the entropy term induces a deviation of $u^\eta$ from this minimum.
To mitigate this deviation, $\eta$ should be chosen such that it exhibits a steep increase away from its minimum.

Motivated by this consideration, we consider a family of polynomials for $\eta$ of the following form:
\begin{align}
    \eta(u) = \left( \frac{u-\mu}{\Delta} \right)^{2n}, \label{eq:eta_even}
\end{align}
where $\mu, \Delta \in \mathbb{R}$, and $n \in \mathbb{Z}_{> 0}$ are parameters that determine the equilibrium state described by the generalized ensemble~\footnote{In particular, in the limit $n\to\infty$, this ensemble approaches the microcanonical ensemble with an energy shell $[\mu-\Delta, \mu+\Delta]$.}.
This choice ensures that $\eta$ grows rapidly when $u$ deviates from $\mu$ by more than $\Delta$.
We show that in this family of generalized ensembles,
there exist ensembles that asymptotically achieve the optimal scaling given by Eq.~\eqref{eq:optimal-scaling_N} for any system at any temperature.

Suppose that the inverse temperature of the target equilibrium state is $\beta$.
By substituting Eq.~\eqref{eq:eta_even} into the temperature formula~\eqref{eq:temperature-formula}, we find that $\mu$ must satisfy
\begin{align}
    \mu = u^\eta - \Delta \left( \frac{\Delta\beta}{2n} \right)^{\frac{1}{2n-1}}, \label{eq:eta_even-mu}
\end{align}
with $\Delta$ and $n$ remaining as free parameters.
Conversely, if Eq.~\eqref{eq:eta_even-mu} holds, the generalized ensemble correctly describes an equilibrium state at the inverse temperature $\beta$.
Using this relation, we can evaluate $\eta(u^\eta) - \eta_\mathrm{min}$.
Since $\eta_\mathrm{min} = 0$, substituting Eq.~\eqref{eq:eta_even-mu} into Eq.~\eqref{eq:eta_even} yields
\begin{align}
    \eta(u^\eta) - \eta_\mathrm{min}
    = \left( \frac{\Delta\beta}{2n} \right)^{\frac{2n}{2n-1}}.
\end{align}
This expression shows that $\eta(u^\eta) - \eta_\mathrm{min}$ can be made arbitrarily small in the limit $n\to\infty$ or $\Delta\to0$.
Consequently, the ensemble-dependent part of the leading factor of the computational cost, $B_N^\eta$, can be reduced arbitrarily by choosing a sufficiently large $n$ or a sufficiently small $\Delta$.
It is important to emphasize that, in these limits,
while the exponent of $B_N^\eta$ can be made small,
the subleading factors become large, as is evident from Eq.~\eqref{eq:scaling}.
Therefore, for finite-size systems, these limits do not yield an optimal ensemble.
We will examine this in more detail in Section~\ref{sec:cost_finite}.

\subsection{Comparison with the Existing Method Based on Generalized Ensembles}
Ref.~\cite{Mizukami2023} employs the generalized ensemble associated with $\eta$, defined as
\begin{align}
    \eta(u) = - 2 \kappa \log (l - u), \label{eq:eta_poly}
\end{align}
where $l$ is a constant greater than or equal to the maximum eigenvalue of $H_N/N$,
and $\kappa$ is a positive constant~\cite{Gerling1993,Sugiura2012}.
This ensemble was originally introduced in the context of classical simulations,
and its use significantly simplifies algorithms based on thermal pure quantum states~\cite{Sugiura2012}.

Although this $\eta$ is not a polynomial, the identity
\begin{align}
    e^{- \frac{1}{2} N \eta(H_N/N)} = (l - H_N/N)^{N\kappa}
\end{align}
allows for a direct implementation of a block-encoding of $\sqrt{\rho_N^\eta}$ via QSVT when $\kappa = k/N$ with $k \in \mathbb{Z}_{\geq 0}$, without requiring a polynomial approximation of the exponential function, as our method does.
These values of $\kappa$ correspond to equilibrium states located at discrete points in the thermodynamic state space,
but these points become dense in the thermodynamic limit~\cite{Sugiura2012,Yoneta2019}.
Therefore, it is sufficient to perform computations only at these discrete values of $\kappa$ to obtain all the thermodynamic properties of the system.

However, as shown in Section~\ref{sec:asymptotic-analysis}, the overall computational cost is dominated by the amplitude-amplification step.
Consequently, this advantage does not affect the scaling of the total computational cost in the thermodynamic limit.
In a similar manner as in Sections~\ref{sec:algorithm} and \ref{sec:asymptotic-analysis},
we can show that the number of queries to the block-encoding $U_H$ of the Hamiltonian
required to construct the generalized ensemble associated with this $\eta$ scales as
\begin{align}
    (\text{number of queries to $U_H$})
    &= \sqrt{\frac{2^N}{e^{N \s(u^\eta)}}}
    \times \left( \frac{l + 1}{l - u^\eta} \right)^{N \kappa}
    \times e^{o(N)}.
\end{align}
For simplicity, we assume here that the maximum and minimum eigenvalues of $H_N/N$ are $+1$ and $-1$, respectively, and that we have access to an $(N, a, 0)$-block-encoding of the Hamiltonian.
From the temperature formula~\eqref{eq:temperature-formula},
we find that for the generalized ensemble to describe the equilibrium state at the inverse temperature $\beta$,
the parameter $\kappa$ must satisfy
\begin{align}
    \kappa = \frac{1}{2} \beta (l - u^\eta).
\end{align}
Substituting this yields
\begin{align}
    (\text{number of queries to $U_H$})
    &= \sqrt{\frac{2^N}{e^{N \s(u^\eta)}}}
    \times \left( \frac{l + 1}{l - u^\eta} \right)^{\frac{1}{2} N \beta (l - u^\eta)}
    \times e^{o(N)}.
\end{align}
This is minimized when $l = 1$, in which case the number of queries is governed only by quantities that are uniquely determined by the target equilibrium state:
\begin{align}
    (\text{number of queries to $U_H$})
    &= \sqrt{\frac{2^N}{e^{N \s(u^\eta)}}}
    \times \exp \left[ \frac{1}{2} N \beta (1 - u^\eta) \log \frac{2}{1 - u^\eta} + o(N) \right].
\end{align}
This implies that the total computational cost grows exponentially with an exponent larger than that of the optimal scaling~\eqref{eq:optimal-scaling_N} achieved by our method.

As seen above, ensembles that are advantageous for classical simulations are not necessarily suitable for quantum simulations.
Therefore, it is important to reexamine and identify ensembles tailored to the quantum computation.

\subsection{Origin of the Advantage Provided by Generalized Ensembles}
At the end of this section,
we explain the underlying mechanism behind the speedup
achieved by our algorithm through the use of generalized ensembles.
This reveals that the efficiency gained by employing generalized ensembles
(or, equivalently, the overhead incurred when using the canonical ensemble)
stems from constraints intrinsic to quantum computations.
To this end, we first revisit our thermal state preparation algorithm described in Section~\ref{sec:algorithm}
and then discuss how the choice of ensemble influences the computational cost. 

The maximally entangled state between the system and ancilla registers prepared in Step~\hyperref[sec:step1]{1} can be expressed in terms of the energy eigenstates $\ket{E_\nu}$ of $H_N$ as
\begin{align}
    \ket{\Psi_0}
    \propto \sum_{\bm{\sigma}} \ket{\bm{\sigma}} \otimes \ket{\bm{\sigma}}
    = \sum_{\nu=1}^{2^N} \ket{E_\nu} \otimes \ket{j_\nu},
\end{align}
where $\ket{\bm{\sigma}} = \bigotimes_{n=1}^N \ket{\sigma_n}_n \ (\sigma_n=\pm 1)$ are spin product states,
and $\ket{j_\nu} = \sum_{\bm{\sigma}} \braket{E_\nu|\bm{\sigma}} \ket{\bm{\sigma}}$ form an orthonormal basis of the ancilla register.
Thus, $\ket{\Psi_0}$ constitutes an equal-weight superposition of all energy eigenstates, each tensored with a mutually orthogonal ancilla state.

Noting that the number of energy eigenstates with the energy density $u$ is given by $\sim e^{N \s(u)}$ according to Boltzmann's entropy formula
and that the gradient of the entropy, $\s'(u)$, coincides with the inverse temperature $\beta$,
we find that the number of energy eigenstates increases exponentially as the temperature increases.
As a result, $\ket{\Psi_0}$ places the majority of its weight on eigenstates lying in the infinite-temperature regime,
and thus represents an infinite-temperature state (purple curve in Fig.~\ref{fig:distribution}).
Therefore, the goal of thermal state preparation is to extract and amplify the contributions of eigenstates within the target energy region from $\ket{\Phi_0}$ so that they become dominant in the resulting state.
In our algorithm, this is accomplished by applying the operator $e^{- \frac{1}{2} N \eta(H_N/N)}$,
which is embedded in the quantum circuit $V$ in Step~\hyperref[sec:step2]{2} (green curve in Fig.~\ref{fig:distribution}),
followed by amplitude amplification in Step~\hyperref[sec:step3]{3} (light blue curve in Fig.~\ref{fig:distribution}).

When the canonical ensemble is employed (i.e., for $\eta=\beta u$),
applying $e^{- \frac{1}{2} \beta H_N}$ amounts to multiplying the amplitudes of the wave function
by an exponentially decaying filter, $e^{-\frac{1}{2} N \beta u}$,
thereby suppressing high-temperature components and emphasizing finite-temperature contributions.
However, as explained in Section~\ref{sec:block-encoding},
quantum computers require all operations to be implemented as unitary operators.
This constraints limits us to embedding operators with a spectral norm not exceeding unity.
Consequently, it is necessary to rescale $e^{- \frac{1}{2} \beta H_N}$
so that all eigenvalues are at most $1$.
Since the maximum eigenvalue of $e^{- \frac{1}{2} \beta H_N}$ is equal to the Boltzmann weight of the ground state, $e^{- \frac{1}{2} N \beta u_0}$,
it must be rescaled so that this value is less than or equal to $1$.
As a result, the filter function applied to the energy eigenstates with eigenvalues $N u$ is at most $e^{- \frac{1}{2} N \beta (u - u_0)}$.
Since $u - u_0 = \Theta(N^0)$ for $u$ in the finite‐temperature energy region,
the expansion coefficients with respect to the energy eigenstates
are multiplied by an exponentially small factor
within the target energy region, as illustrated in Fig.~\ref{fig:filter}.
This causes the computational cost for amplitude amplification in Step~\hyperref[sec:step3]{3} to become exponentially larger than the intuitive estimate of $\sqrt{2^N / e^{N \s(u)}}$ (see Section~\ref{sec:asymptotic-analysis}).

Ideally, the filter function should be designed to be a maximum value of $1$ within the target energy region.
By using generalized ensembles, the filter function can be chosen with considerable flexibility.
This flexibility allows us to construct a filter function
whose maximum does not exceed $1$, as required for embedding into a unitary operator,
while ensuring that it approaches this maximum arbitrarily closely
within the target energy region (see Fig.~\ref{fig:filter}).
This significantly reduces the computational cost associated with the amplitude amplification process.

In summary, the poor compatibility between the canonical ensemble and the block-encoding technique arises from an intrinsic constraint of quantum computation and incurs an exponential overhead in thermal-state preparation.
In contrast, the flexibility in choosing generalized ensembles
allows us to mitigate this overhead, significantly reducing total computational cost.

\section{Computational Cost for Finite-Size Systems} \label{sec:cost_finite}
In the previous section, using the generalized ensemble associated with $\eta$ defined by Eq.~\eqref{eq:eta_even},
we have shown that taking the limit $n \to \infty$ or $\delta \to 0$ allows one 
to optimize the scaling in the leading factor of the computational cost.
However, as mentioned in Section~\ref{sec:optimal-ensemble},
such ensemble is not optimal for finite-size systems.
By using the expression for the computational cost provided in Theorem~\ref{theorem:main},
we can optimize the computational cost even for finite systems.
In this section, by precisely evaluating the computational cost for specific models,
we demonstrate that even for relatively small finite-size systems,
which are likely to be the first applications of this algorithm in the future, 
one can significantly reduce the computational cost using our method.

\begin{figure*}[t]
    \begin{minipage}{0.48\linewidth}
        \centering
        \includegraphics[height=6cm]{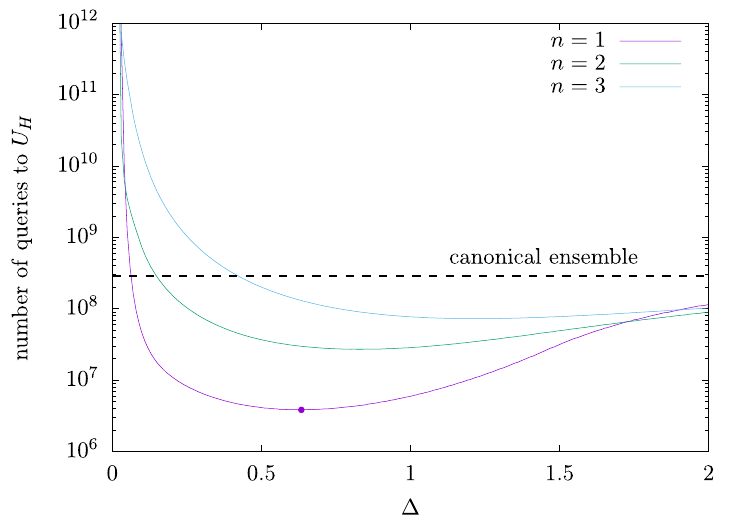}
        \subcaption{}
        \label{fig:optimization_L50}
    \end{minipage}
    \hfill
    \begin{minipage}{0.48\linewidth}
        \centering
        \includegraphics[height=6cm]{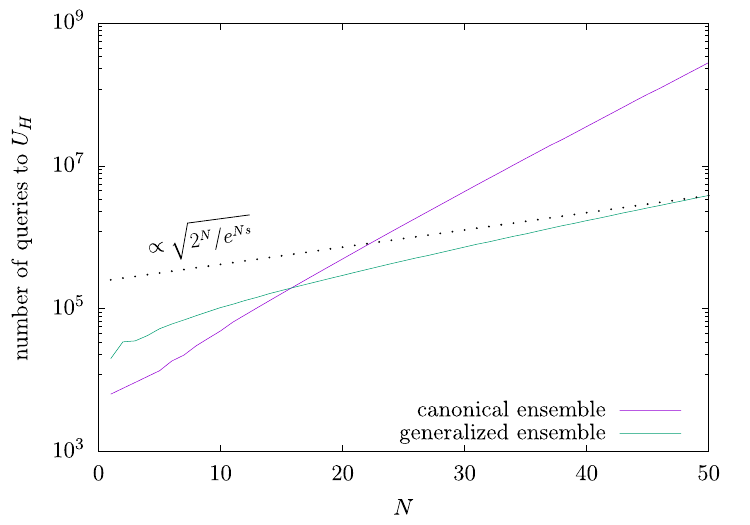}
        \subcaption{}
        \label{fig:optimal_L50}
    \end{minipage}
    \caption{%
        (\subref{fig:optimization_L50}) Number of queries to the block-encoding of the Hamiltonian, $U_H$, required to construct the generalized ensemble associated with $\eta$ defined by Eq.~\eqref{eq:eta_even} for $N=50$, shown as a function of the ensemble parameters.
        The ensembles are compared under the constraint that they describe the same equilibrium state, with the inverse temperature fixed at $\beta = 0.5$.
        (\subref{fig:optimal_L50}) $N$-dependence of the number of queries required to construct the canonical ensemble and the generalized ensemble with optimized ensemble parameters for $\beta=0.5$.
        The dashed line indicates the optimal scaling $\sqrt{2^N / e^{N\s(u^\eta)}}$ predicted by the asymptotic analysis.
    }
\end{figure*}

As an illustration, we consider the example of free spins, whose Hamiltonian is given by
\begin{align}
    H_N = \sum_{n=1}^{N} Z_{n}.
    \label{eq:Hamiltonian_free-spin}
\end{align}
For this system, statistical-mechanical quantities can be computed exactly
not only for the canonical ensemble
but also for the generalized ensemble (see Appendix~\ref{sec:free-spins}).
This allows for an explicit evaluation of the computational cost for large systems
using the expression provided in Theorem~\ref{theorem:main}.
Although free spins are chosen for the sake of analytical tractability,
we emphasize that, as indicated by Theorem~\ref{theorem:main},
the number of queries to the block-encoding of the Hamiltonian depends on the system
solely through its partition function
and is independent of microscopic properties such as integrability.
Hence, the simplicity of the model does not undermine our analysis.
In what follows, we assume that the Hamiltonian is provided through a $(\|H_N\|,a,0)$-block-encoding $U_H$.

As a concrete scenario, consider a quantum computer capable of simulating systems of up to $50$ sites is used.
In this case, the computation proceeds as follows:
First, we determine the ensemble parameters that minimize the computational cost for the most computationally expensive $50$-site system.
Next, we fix these ensemble parameters
and perform simulations for various system sizes up to $50$ sites.
Finally, we extrapolate the results to the thermodynamic limit.

In Fig.~\ref{fig:optimization_L50}, we plot the number of queries to $U_H$ required to construct the generalized ensemble associated with $\eta$ defined by Eq.~\eqref{eq:eta_even} for $N=50$ as a function of the ensemble parameters.
To ensure a fair comparison, we compare ensembles that describe the same equilibrium state.
Since this model does not exhibit a first-order phase transition and its equilibrium state is therefore uniquely determined by the inverse temperature, it suffices to compare ensembles at the same inverse temperature.
Then we impose the constraint~\eqref{eq:eta_even-mu} on the ensemble parameters so that the inverse temperature is fixed at $\beta=0.5$.
We find that the computational cost is minimized at $n = 1$ and $\Delta \simeq 0.63$ and is reduced by two orders of magnitude compared to the case using the canonical ensemble.
It can also be seen that, even when using generalized ensembles, an inappropriate choice of parameters can result in a higher computational cost than that of the canonical ensemble.
This highlights the importance of not only using generalized ensembles, but also carefully selecting one that is tailored to the problem.

Figure~\ref{fig:optimal_L50} shows the $N$-dependence of the number of queries to $U_H$ required to construct the canonical ensemble and the generalized ensembles with the optimized parameters ($n = 1, \Delta \simeq 0.63$) identified above.
We can confirm that the cost for preparing the generalized ensemble is exponentially smaller than that for the canonical ensemble.

To clearly highlight the difference in asymptotic behavior,
we plot the $N$-dependence of the number of queries to $U_H$ for the ensemble optimized at $N=1000$
in Fig.~\ref{fig:optimal_L1000}.
The ensemble parameters are chosen to minimize the number of queries for the system with $N=1000$ under the constraint $\beta=0.5$.
We observe that the optimal scaling in the thermodynamic limit, $\sqrt{2^N / e^{N\s(u^\eta)}}$, predicted by our asymptotic analysis, is indeed nearly achieved.
This confirms the validity of our analysis in Section~\ref{sec:cost_TDL}.

\begin{figure}[t]
    \centering
    \includegraphics[height=6cm]{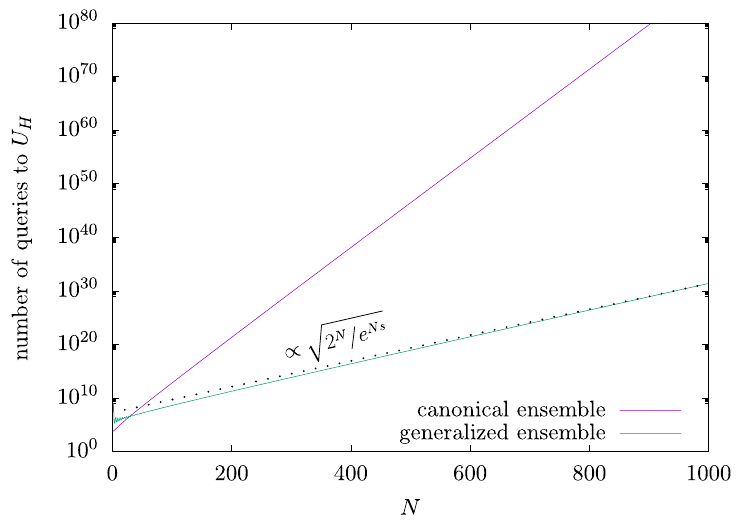}
    \caption{%
        $N$-dependence of the number of queries to the block-encoding of the Hamiltonian, $U_H$, required to construct the canonical ensemble and the generalized ensemble associated with $\eta$ defined by Eq.~\eqref{eq:eta_even} for $\beta=0.5$.
        In the generalized ensemble, the ensemble parameters are optimized to minimize the number of queries for the system with $N=1000$.
        The dashed line indicates the optimal scaling $\sqrt{2^N / e^{N\s(u^\eta)}}$ predicted by the asymptotic analysis.
    }
    \label{fig:optimal_L1000}
\end{figure}

In summary, by employing appropriately chosen generalized ensembles, the scaling of the computational cost of thermal state preparation can be improved, yielding significant reductions even for moderately sized systems.

\section{Conclusion} \label{sec:conclusion}
In this work, we have proposed a quantum algorithm for thermal state preparation
based on generalized statistical ensembles
within the framework of quantum singular value transformation.
Through a detailed analysis of the computational cost,
we have shown that by appropriately choosing the ensemble,
one can significantly reduce the computational cost
compared to existing methods based on the canonical ensemble,
not only in terms of asymptotic scaling in the thermodynamic limit,
but also for moderately sized finite systems
that are likely to be the first targets of applications.
We have further shown that the origin of this advantage stems from fundamental constraints of quantum computation,
which inherently limit the efficiency of the algorithm based on the canonical ensemble. 
Since our algorithm is applicable to arbitrary thermodynamic systems at any temperature,
it provides a general and efficient method for finite-temperature quantum simulation.
More broadly, our results highlight the potential of ensemble design
as a powerful tool for enhancing the efficiency of a broad range of quantum algorithms.

So far, we have studied the advantages of employing generalized ensembles
from the perspective of reducing the computational cost.
However, the benefits of generalized ensembles are not limited to computational efficiency,
particularly when investigating systems that exhibit first-order phase transitions.
Let us briefly explain this.

As mentioned in Section~\ref{sec:ensemble},
different ensembles are thermodynamically equivalent,
provided they are appropriately chosen.
However, the canonical ensemble is inappropriate in the first-order phase-transition region~\cite{Gross2001,Yoneta2019},
where multiple phases coexist in various proportions.
To illustrate this point, let us consider a system which undergoes a temperature-driven first-order phase transition.
In a phase coexistence state, the temperature takes the same value between the coexisting phases and therefore cannot distinguish between them.
Consequently, the canonical ensemble, which is specified by the temperature, can represent only a particular state among the various states in the phase transition region.
In particular, under periodic boundary conditions, it yields either a single-phase state or a statistical mixture of single-phase states, and fails to give a phase coexistence state~\cite{Yoneta2019}.

In contrast, since coexisting phases can be distinguished by their energy, by employing an ensemble in which the energy has a macroscopically definite value, one can obtain each of the equilibrium states within the first-order phase-transition region.
Since generalized ensembles ensure that the energy takes a macroscopically definite value, we can investigate microscopic structures of these states, including phase coexistence states.
Since our thermal state preparation algorithm is based on generalized ensembles,
it is directly applicable to systems exhibiting first-order phase transitions.

\begin{acknowledgments}
We thank K.~Fujii, T.~Ikeda, K.~Mizuta, and A.~Shimizu for discussions.
This work was supported by RIKEN Special Postdoctoral Researcher Program.
\end{acknowledgments}

\onecolumngrid

\appendix
\section{Quantum Circuit for Eigenvalue Transformation} \label{sec:EVT}
In this appendix, we provide an overview of the construction presented in Ref.~\cite{Gilyen2019}
for the quantum circuit that implements a block-encoding of the eigenvalue transformation of a Hermitian matrix,
as stated in Theorem~\ref{theorem:EVT}.

Define two real polynomials of definite parity,
\begin{align}
    P^{(0)}_R(x) = P(x) + P(-x), \quad
    P^{(1)}_R(x) = P(x) - P(-x).
\end{align}
From the condition~\eqref{eq:EVT_condition} of the theorem,
for each $c=0,1$, the polynomial $P_R^{(c)}(x) = P(x) + (-1)^c P(-x)$ satisfies conditions~\ref{cond:QSP_real_1} and \ref{cond:QSP_real_2} of Theorem~\ref{theorem:QSP_real}.
Therefore, there exist polynomials $P^{(c)} \in \mathbb{C}[x]$
such that $P_R^{(c)}(x) = \mathrm{Re}[P^{(c)}(x)]$
and satisfy conditions~\ref{cond:QSP_1}-\ref{cond:QSP_4} of Theorem~\ref{theorem:QSP}.
Then let $\Phi^{(c)}$ denote the QSP phases corresponding to $P^{(c)}$.

It is straightforward to verify that QSP for $-\Phi^{(c)}$ generates ${P^{(c)}}^*$.
Therefore, a linear combination of the alternating phase modulation sequences (see Theorem~\ref{theorem:QSVT} for the definition),
\begin{align}
    \frac{1}{2} \sum_{b=0,1} U_{(-1)^b \Phi^{(c)}},
\end{align}
encodes $P_R^{(c)}(A/\alpha)$ for $c = 0, 1$.
Consequently,
\begin{align}
    \frac{1}{4} \sum_{b,c=0,1} U_{(-1)^b \Phi^{(c)}}
\end{align}
encodes $P(A/\alpha)$.

To implement this as a quantum circuit efficiently, we utilize the identity~\eqref{eq:exp-i-phi-Pi} and the linear combination of unitaries technique (see Fig.\ref{fig:exp-i-phi-Pi_2}).
As a result, the quantum circuit for a block-encoding of the eigenvalue transformation $P(A/\alpha)$ can be implemented as shown in Fig.~\ref{fig:EVT}.
The resulting circuit uses
$d-1$ queries to $U$ and $U^\dagger$,
a single query to either controlled-$U$ or controlled-$U^\dagger$,
and $\Order((a+1)d)$ additional one- and two-qubit gates.

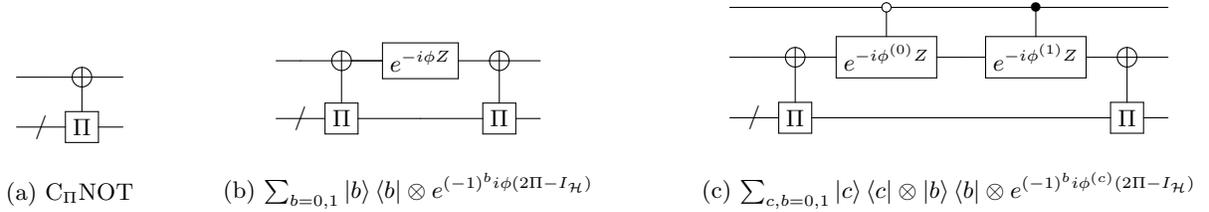
\begin{figure*}[h]
    \begin{minipage}[b]{0.19\linewidth}
        \centering
        \begin{displaymath}
            \Qcircuit @C=1.em @R=1.em {
                & \qw &\targ{1} & \qw\\
                & /\qw & \gate{\Pi}\qwx[-1] & \qw
            }
        \end{displaymath}
        \subcaption{$\cnot{\Pi}$}
        \label{fig:CPiNOT}
    \end{minipage}
    \begin{minipage}[b]{0.3\linewidth}
        \centering
        \begin{displaymath}
            \Qcircuit @C=1.em @R=1.em {
                & \qw & \targ{1} & \gate{e^{- i \phi Z}} \qw & \targ{1} & \qw\\
                & /\qw & \gate{\Pi}\qwx & \qw & \gate{\Pi}\qwx & \qw
            }
        \end{displaymath}
        \subcaption{$\sum_{b=0,1} \ket{b}\bra{b} \otimes e^{(-1)^b i \phi (2 \Pi - \I_\mathcal{H})}$}
        \label{fig:exp-i-phi-Pi_1}
    \end{minipage}
    \begin{minipage}[b]{0.49\linewidth}
        \centering
        \begin{displaymath}
            \Qcircuit @C=1.em @R=1.em {
                & \qw & \qw & \ctrlo{1} & \qw & \ctrl{1} & \qw & \qw\\
                & \qw & \targ{1} & \gate{e^{- i \phi^{(0)} Z}} & \qw & \gate{e^{- i \phi^{(1)} Z}} & \targ{1} & \qw\\
                & /\qw &\gate{\Pi}\qwx & \qw & \qw & \qw & \gate{\Pi}\qwx & \qw
            }
        \end{displaymath}
        \subcaption{$\sum_{c,b=0,1} \ket{c}\bra{c} \otimes \ket{b}\bra{b} \otimes e^{(-1)^b i \phi^{(c)} (2 \Pi - \I_\mathcal{H})}$}
        \label{fig:exp-i-phi-Pi_2}
    \end{minipage}
    \caption{%
        Gate sequences used for quantum singular value transformation and eigenvalue transformation.
    }
\end{figure*}

\begin{turnpage}
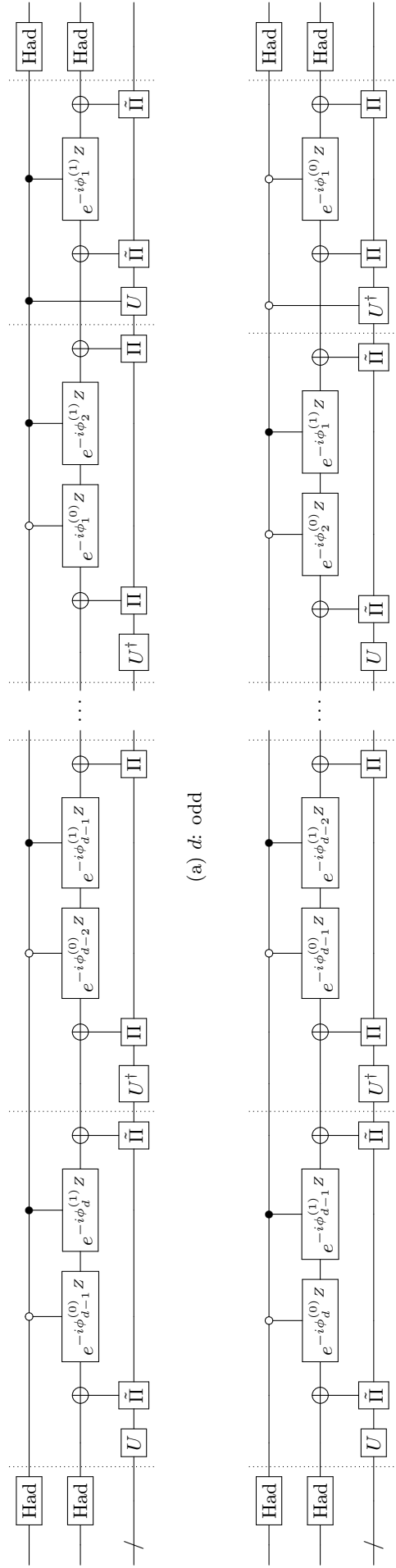
\begin{figure}
    {
        \centering
        \begin{align*}
            \Qcircuit @C=1.em @R=1.em {
                & \qw & \gate{\hadamard} \ar@{.}[]+<1.7em,1em>;[dd]+<1.7em,-1em>
                & \qw & \qw & \ctrlo{+1} & \ctrl{+1} & \qw \ar@{.}[]+<1.2em,1em>;[dd]+<1.2em,-1em>
                & \qw & \qw & \ctrlo{+1} & \ctrl{+1} & \qw \ar@{.}[]+<1.2em,1em>;[dd]+<1.2em,-1em>
                & \qw & & \ar@{.}[]+<0.4em,1em>;[dd]+<0.4em,-1em>
                & \qw & \qw & \ctrlo{+1} & \ctrl{+1} & \qw \ar@{.}[]+<1.2em,1em>;[dd]+<1.2em,-1em>
                & \ctrl{+2} & \qw & \ctrl{+1} & \qw \ar@{.}[]+<1.2em,1em>;[dd]+<1.2em,-1em>
                & \gate{\hadamard} & \qw\\
                & \qw & \gate{\hadamard}
                & \qw & \targ & \gate{e^{- i \phi_{d-1}^{(0)} Z}} & \gate{e^{- i \phi_{d}^{(1)} Z}} & \targ
                & \qw & \targ & \gate{e^{- i \phi_{d-2}^{(0)} Z}} & \gate{e^{- i \phi_{d-1}^{(1)} Z}} & \targ
                & \qw & \cdots &
                & \qw & \targ & \gate{e^{- i \phi_{1}^{(0)} Z}} & \gate{e^{- i \phi_{2}^{(1)} Z}} & \targ
                & \qw & \targ & \gate{e^{- i \phi_{1}^{(1)} Z}} & \targ
                & \gate{\hadamard} & \qw\\
                & /\qw & \qw
                & \gate{U} & \gate{\tilde{\Pi}}\qwx{-1} & \qw & \qw & \gate{\tilde{\Pi}}\qwx{-1}
                & \gate{U^\dagger} & \gate{\Pi}\qwx{-1} & \qw & \qw & \gate{\Pi}\qwx{-1}
                & \qw & &
                & \gate{U^\dagger} & \gate{\Pi}\qwx{-1} & \qw & \qw & \gate{\Pi}\qwx{-1}
                & \gate{U} & \gate{\tilde{\Pi}}\qwx{-1} & \qw & \gate{\tilde{\Pi}}\qwx{-1}
                & \qw & \qw
            }
        \end{align*}
        \subcaption{$d$: odd}
        \label{fig:EVT_odd}
    }
    {
        \centering
        \begin{align*}
            \Qcircuit @C=1.em @R=1.em {
                & \qw & \gate{\hadamard} \ar@{.}[]+<1.7em,1em>;[dd]+<1.7em,-1em>
                & \qw & \qw & \ctrlo{+1} & \ctrl{+1} & \qw \ar@{.}[]+<1.2em,1em>;[dd]+<1.2em,-1em>
                & \qw & \qw & \ctrlo{+1} & \ctrl{+1} & \qw \ar@{.}[]+<1.2em,1em>;[dd]+<1.2em,-1em>
                & \qw & & \ar@{.}[]+<0.4em,1em>;[dd]+<0.4em,-1em>
                & \qw & \qw & \ctrlo{+1} & \ctrl{+1} & \qw \ar@{.}[]+<1.2em,1em>;[dd]+<1.2em,-1em>
                & \ctrlo{+2} & \qw & \ctrlo{+1} & \qw \ar@{.}[]+<1.2em,1em>;[dd]+<1.2em,-1em>
                & \gate{\hadamard} & \qw\\
                & \qw & \gate{\hadamard}
                & \qw & \targ & \gate{e^{- i \phi_{d}^{(0)} Z}} & \gate{e^{- i \phi_{d-1}^{(1)} Z}} & \targ
                & \qw & \targ & \gate{e^{- i \phi_{d-1}^{(0)} Z}} & \gate{e^{- i \phi_{d-2}^{(1)} Z}} & \targ
                & \qw & \cdots &
                & \qw & \targ & \gate{e^{- i \phi_{2}^{(0)} Z}} & \gate{e^{- i \phi_{1}^{(1)} Z}} & \targ
                & \qw & \targ & \gate{e^{- i \phi_{1}^{(0)} Z}} & \targ
                & \gate{\hadamard} & \qw\\
                & /\qw & \qw
                & \gate{U} & \gate{\tilde{\Pi}}\qwx{-1} & \qw & \qw & \gate{\tilde{\Pi}}\qwx{-1}
                & \gate{U^\dagger} & \gate{\Pi}\qwx{-1} & \qw & \qw & \gate{\Pi}\qwx{-1}
                & \qw & &
                & \gate{U} & \gate{\tilde{\Pi}}\qwx{-1} & \qw & \qw & \gate{\tilde{\Pi}}\qwx{-1}
                & \gate{U^\dagger} & \gate{\Pi}\qwx{-1} & \qw & \gate{\Pi}\qwx{-1}
                & \qw & \qw
            }
        \end{align*}
        \subcaption{$d$: even}
        \label{fig:EVT_even}
    }
    \caption{%
        Quantum circuits implementing eigenvalue transformation
        for (\subref{fig:EVT_odd}) odd-degree
        and (\subref{fig:EVT_even}) even-degree polynomials.
    }
    \label{fig:EVT}
\end{figure}
\end{turnpage}

\section{Proof of Theorem~\ref{theorem:FPAA}} \label{sec:proof_FPAA}
For the sake of providing a precise evaluation of the computational cost, we modify the proof of theorems from Ref.~\cite{Gilyen2019}.
In the proof, the following lemma plays a crucial role:
\begin{lemma}[Polynomial approximation of the error function $\erf(k x)$~\cite{Low2017_1}] \label{lemma:error-function}
For all $k > 0$, the polynomial $p_{\erf, k, n}(x)$ of degree $n$
\begin{align}
    p_{\erf, k, n}(x)
    = \frac{2 k e^{- k^2 / 2}}{\sqrt{\pi}} \left(
        I_0(k^2 / 2) x
        + \sum_{j=1}^{(n-1)/2} I_j(k^2 / 2) (-1)^j \left( \frac{T_{2j+1}(x)}{2j+1} - \frac{T_{2j-1}(x)}{2j-1} \right)
    \right)
\end{align}
satisfies
\begin{align}
    \max_{x \in [-1, +1]}
    \left|
        p_{\erf, k, n}(x) - \erf(k x)
    \right|
    \leq \frac{4 k}{\sqrt{\pi} n} \left(
        2 e^{- (n-1)^2 / 8 t} + e^{- k^2 / 2 - t}
    \right)
\end{align}
for every integer $t \geq e^2 k^2 / 2$.
Here, $I_j(x)$ are modified Bessel functions of the first kind, and $T_j(x)$ are Chebyshev polynomials of the first kind.
\end{lemma}

\begin{proof}[Proof of Theorem~\ref{theorem:FPAA}]
Since
\begin{align}
    \left|
        \erf(k x) - \sgn(x)
    \right|
    &= \frac{2}{\sqrt{\pi}} \int_{k |x|}^{\infty} d{x'}\ e^{-{x'}^2}
    \leq \frac{2}{\sqrt{\pi}} \int_{k |x|}^{\infty} d{x'}\ \frac{x'}{k |x|} e^{-{x'}^2}
    = \frac{1}{\sqrt{\pi} k |x|} e^{- (k x)^2},
\end{align}
the polynomial $p_{\erf, k, n}$ in Lemma~\ref{lemma:error-function} satisfies
\begin{align}
    \max_{x \in [-1, -\delta] \cup [+\delta, +1]}
    \left|
        p_{\erf, k, n}(x) - \sgn(x)
    \right|
    \leq \frac{1}{\sqrt{\pi} k \delta} e^{- (k \delta)^2}
    + \frac{4 k}{\sqrt{\pi} n} \left(
        2 e^{- (n-1)^2 / 8 t} + e^{- k^2 / 2 - t}
    \right).
\end{align}
Thus, by choosing
\begin{align}
    k = \frac{1}{\delta} \sqrt{\frac{1}{2} W\left( \frac{2^{11}}{\pi \epsilon^8} \right)}, \quad
    t = \left\lceil \max \left\{
        e^2 k^2 / 2, \log \frac{2^8 k}{\sqrt{\pi} \epsilon^4}
    \right\} \right\rceil, \quad
    d = 2 \left\lceil \sqrt{t W \left( \frac{2^{16} k^2}{\pi t \epsilon^8} \right)} \right\rceil + 1,
\end{align}
it is ensured that
\begin{align}
    \left|
        p_{\erf, k, n}(x) - \sgn(x)
    \right|
    \leq \epsilon^4/32 + \epsilon^4/64 + \epsilon^4/64
    = \epsilon^4/16.
\end{align}
For this $p_{\erf, k, n}$, we define a degree-$d$ odd real polynomial $P_R$ as
\begin{align}
    P_R(x) = \frac{1}{1+\epsilon^4/16} p_{\erf, k, n}(x).
\end{align}
Then this satisfies
\begin{align}
    \max_{x \in [-1, -\delta] \cup [+\delta, +1]} \left| P_R(x) - \sgn(x) \right| &\leq \epsilon^4/8, \label{eq:sign-function}
\end{align}
and
\begin{align}
    | P_R(x) | &\leq 1.
\end{align}
By Theorem~\ref{theorem:QSP_real}, there exists a complex polynomial $P$ of the same degree as $P_R$ satisfying $\mathrm{Re}[P(x)] = P_R(x)$ and fulfilling conditions~\ref{cond:QSP_1}-\ref{cond:QSP_4} of Theorem~\ref{theorem:QSP}.
Then, by applying Theorem~\ref{theorem:QSVT}, we can construct a quantum circuit $\tilde{U}$ with $d$ queries to (controlled) $U$ and $U^\dagger$ via quantum singular value transformation, such that
\begin{align}
    \Pi \tilde{U} \ket{\psi_0}\bra{\psi_0}
    = P^{(SV)} \left( \Pi U \ket{\psi_0}\bra{\psi_0} \right).
\end{align}
Since $\Pi U \ket{\psi_0} = \alpha \ket{\psi_G}$, we have
\begin{align}
    \Pi U \ket{\psi_0}\bra{\psi_0}
    = \alpha \ket{\psi_G}\bra{\psi_0}.
\end{align}
Thus, singular value transformation yields
\begin{align}
    P^{(SV)} \left( \Pi U \ket{\psi_0}\bra{\psi_0} \right)
    = P(\alpha) \ket{\psi_G}\bra{\psi_0}.
\end{align}
Since $\mathrm{Re}[P(x)]=P_R(x)$ satisfies Eq.~\eqref{eq:sign-function} and the condition~\ref{cond:QSP_2} ($|P(\alpha)| \leq 1$) of Theorem~\ref{theorem:QSP} holds, we have $1-\epsilon^4/8 \leq \mathrm{Re}[P(\alpha)] \leq 1$ and $(\mathrm{Im}[P(\alpha)])^2 \leq 1-(1-\epsilon^4/8)^2$.
Hence, we find
\begin{align}
    {| P(\alpha) - 1 |}^2
    = ( \mathrm{Re}[P(\alpha)] - 1 )^2 + ( \mathrm{Im}[P(\alpha)] )^2
    \leq \epsilon^4/4.
\end{align}
Therefore, we obtain
\begin{align}
    \left\| \ket{\psi_G}\bra{\psi_G} \tilde{U} \ket{\psi_0}\bra{\psi_0} - \ket{\psi_G}\bra{\psi_0} \right\|
    = \left\|
        \left( \Pi \tilde{U} \ket{\psi_0}\bra{\psi_0} - P(\alpha) \ket{\psi_G}\bra{\psi_0} \right)
        + \left( P(\alpha) - 1 \right) \ket{\psi_G}\bra{\psi_0}
    \right\|
    \leq \epsilon^2/2.
\end{align}
On the other hand,
\begin{align}
    {\|\ket{\psi_G}-\tilde{U}\ket{\psi_0}\|}^2
    &= 2 \mathrm{Re} [ 1 - \braket{\psi_G|\tilde{U}|\psi_0} ]
    \leq 2 | 1 - \braket{\psi_G|\tilde{U}|\psi_0} |
    = 2 \left\| \ket{\psi_G}\bra{\psi_G} \tilde{U} \ket{\psi_0}\bra{\psi_0} - \ket{\psi_G}\bra{\psi_0} \right\|.
\end{align}
Therefore, it follows that $\tilde{U}$ satisfies
\begin{align}
    \|\ket{\psi_G}-\tilde{U}\ket{\psi_0}\| \leq \epsilon.
\end{align}
\end{proof}

\section{Proof of Eq.~\eqref{eq:p-exponential_bound}} \label{sec:proof_p-exponential_bound}
In this Appendix, we prove that the polynomial $p_{\exp, \lambda, n}$ in Lemma~\ref{lemma:exp-function} satisfies
\begin{align}
    \max_{x \in [-1, +1]} |p_{\exp, \lambda, n}(x)| \leq 1.
\end{align}

\begin{proof}
From Eq.~\eqref{eq:p_exp}, we find
\begin{align}
    | p_{\exp, \lambda, n}(x) |
    \leq e^{- \lambda} \left(
            I_0(\lambda) + 2 \sum_{j=1}^{n} I_j(\lambda)
    \right)
    < e^{- \lambda} \left(
            I_0(\lambda) + 2 \sum_{j=1}^{\infty} I_j(\lambda)
    \right)
    = e^{- \lambda} \sum_{j=-\infty}^{\infty} I_j(\lambda)
    = 1.
\end{align}
Here, we have used the fact that $I_j(x)$ is positive for $x > 0$
and that $|T_j(x)| \leq 1$ for $x \in [-1, +1]$ in the first and second inequalities, 
the symmetry property $I_j(x) = I_{-j}(x)$ in the third equality,
and the integral representation of $I_j$,
\begin{align}
    I_j(x) = \frac{1}{\pi} \int_{0}^{\pi} d\theta\ e^{x \cos \theta} \cos(j \theta),
\end{align}
in the fourth equality~\cite{Abramowitz1964}.
\end{proof}

\section{Statistical Mechanical Quantities of Free Spins in the Canonical and Generalized Ensembles} \label{sec:free-spins}
In this section, we compute the partition functions of the canonical and generalized ensembles for free spins, which are used for the evaluation of the computational cost of thermal state preparation discussed in Section~\ref{sec:cost_finite}.

First, we calculate the partition function of the canonical ensemble.
From a simple calculation, we have
\begin{align}
    \Z_N^\mathrm{can} = \mathrm{Tr} \left[ e^{-\beta H_N} \right] = \left( 2 \cosh \beta \right)^N.
\end{align}

Next, we calculate the partition function for the generalized ensemble.
The energy eigenstates of free spins are spin basis states, and their eigenvalues are determined by the total magnetization.
Thus, the partition function of the generalized ensemble associated with $\eta$ is given by
\begin{align}
    \Z_N^\eta = \mathrm{Tr} \left[ e^{-N\eta(H_N/N)} \right] = \sum_{k=0}^N \frac{N!}{k!(N-k)!} e^{-N\eta(2k/N-1)}.
\end{align}

\bibliography{document}

\begin{thebibliography}{79}%
\makeatletter
\providecommand \@ifxundefined [1]{%
 \@ifx{#1\undefined}
}%
\providecommand \@ifnum [1]{%
 \ifnum #1\expandafter \@firstoftwo
 \else \expandafter \@secondoftwo
 \fi
}%
\providecommand \@ifx [1]{%
 \ifx #1\expandafter \@firstoftwo
 \else \expandafter \@secondoftwo
 \fi
}%
\providecommand \natexlab [1]{#1}%
\providecommand \enquote  [1]{``#1''}%
\providecommand \bibnamefont  [1]{#1}%
\providecommand \bibfnamefont [1]{#1}%
\providecommand \citenamefont [1]{#1}%
\providecommand \href@noop [0]{\@secondoftwo}%
\providecommand \href [0]{\begingroup \@sanitize@url \@href}%
\providecommand \@href[1]{\@@startlink{#1}\@@href}%
\providecommand \@@href[1]{\endgroup#1\@@endlink}%
\providecommand \@sanitize@url [0]{\catcode `\\12\catcode `\$12\catcode `\&12\catcode `\#12\catcode `\^12\catcode `\_12\catcode `\%12\relax}%
\providecommand \@@startlink[1]{}%
\providecommand \@@endlink[0]{}%
\providecommand \url  [0]{\begingroup\@sanitize@url \@url }%
\providecommand \@url [1]{\endgroup\@href {#1}{\urlprefix }}%
\providecommand \urlprefix  [0]{URL }%
\providecommand \Eprint [0]{\href }%
\providecommand \doibase [0]{https://doi.org/}%
\providecommand \selectlanguage [0]{\@gobble}%
\providecommand \bibinfo  [0]{\@secondoftwo}%
\providecommand \bibfield  [0]{\@secondoftwo}%
\providecommand \translation [1]{[#1]}%
\providecommand \BibitemOpen [0]{}%
\providecommand \bibitemStop [0]{}%
\providecommand \bibitemNoStop [0]{.\EOS\space}%
\providecommand \EOS [0]{\spacefactor3000\relax}%
\providecommand \BibitemShut  [1]{\csname bibitem#1\endcsname}%
\let\auto@bib@innerbib\@empty
\bibitem [{\citenamefont {Feynman}(1982)}]{Feynman1982}%
  \BibitemOpen
  \bibfield  {author} {\bibinfo {author} {\bibfnamefont {R.~P.}\ \bibnamefont {Feynman}},\ }\bibfield  {title} {\bibinfo {title} {{Simulating physics with computers}},\ }\href {https://doi.org/10.1007/BF02650179} {\bibfield  {journal} {\bibinfo  {journal} {Int. J. Theor. Phys.}\ }\textbf {\bibinfo {volume} {21}},\ \bibinfo {pages} {467} (\bibinfo {year} {1982})}\BibitemShut {NoStop}%
\bibitem [{\citenamefont {Lloyd}(1996)}]{Lloyd1996}%
  \BibitemOpen
  \bibfield  {author} {\bibinfo {author} {\bibfnamefont {S.}~\bibnamefont {Lloyd}},\ }\bibfield  {title} {\bibinfo {title} {{Universal Quantum Simulators}},\ }\href {https://doi.org/10.1126/science.273.5278.1073} {\bibfield  {journal} {\bibinfo  {journal} {Science}\ }\textbf {\bibinfo {volume} {273}},\ \bibinfo {pages} {1073} (\bibinfo {year} {1996})}\BibitemShut {NoStop}%
\bibitem [{\citenamefont {Sly}(2010)}]{Sly2010}%
  \BibitemOpen
  \bibfield  {author} {\bibinfo {author} {\bibfnamefont {A.}~\bibnamefont {Sly}},\ }\bibfield  {title} {\bibinfo {title} {{Computational Transition at the Uniqueness Threshold}},\ }in\ \href {https://doi.org/10.1109/FOCS.2010.34} {\emph {\bibinfo {booktitle} {2010 IEEE 51st Annual Symposium on Foundations of Computer Science}}}\ (\bibinfo  {publisher} {IEEE Computer Society},\ \bibinfo {year} {2010})\ pp.\ \bibinfo {pages} {287--296}\BibitemShut {NoStop}%
\bibitem [{\citenamefont {Sly}\ and\ \citenamefont {Sun}(2012)}]{Sly2012}%
  \BibitemOpen
  \bibfield  {author} {\bibinfo {author} {\bibfnamefont {A.}~\bibnamefont {Sly}}\ and\ \bibinfo {author} {\bibfnamefont {N.}~\bibnamefont {Sun}},\ }\bibfield  {title} {\bibinfo {title} {{The Computational Hardness of Counting in Two-Spin Models on $d$-Regular Graphs}},\ }in\ \href {https://doi.org/10.1109/FOCS.2012.56} {\emph {\bibinfo {booktitle} {2013 IEEE 54th Annual Symposium on Foundations of Computer Science}}}\ (\bibinfo  {publisher} {IEEE Computer Society},\ \bibinfo {address} {Los Alamitos, CA, USA},\ \bibinfo {year} {2012})\ pp.\ \bibinfo {pages} {361--369}\BibitemShut {NoStop}%
\bibitem [{\citenamefont {Galanis}\ \emph {et~al.}(2016)\citenamefont {Galanis}, \citenamefont {^^c5^^a0tefankovi^^c4^^8d},\ and\ \citenamefont {Vigoda}}]{Galanis2016}%
  \BibitemOpen
  \bibfield  {author} {\bibinfo {author} {\bibfnamefont {A.}~\bibnamefont {Galanis}}, \bibinfo {author} {\bibfnamefont {D.}~\bibnamefont {^^c5^^a0tefankovi^^c4^^8d}},\ and\ \bibinfo {author} {\bibfnamefont {E.}~\bibnamefont {Vigoda}},\ }\bibfield  {title} {\bibinfo {title} {{Inapproximability of the Partition Function for the Antiferromagnetic Ising and Hard-Core Models}},\ }\href {https://doi.org/10.1017/S0963548315000401} {\bibfield  {journal} {\bibinfo  {journal} {Combinator. Probab. Comp.}\ }\textbf {\bibinfo {volume} {25}},\ \bibinfo {pages} {500} (\bibinfo {year} {2016})}\BibitemShut {NoStop}%
\bibitem [{\citenamefont {Poulin}\ and\ \citenamefont {Wocjan}(2009)}]{Poulin2009}%
  \BibitemOpen
  \bibfield  {author} {\bibinfo {author} {\bibfnamefont {D.}~\bibnamefont {Poulin}}\ and\ \bibinfo {author} {\bibfnamefont {P.}~\bibnamefont {Wocjan}},\ }\bibfield  {title} {\bibinfo {title} {{Sampling from the Thermal Quantum Gibbs State and Evaluating Partition Functions with a Quantum Computer}},\ }\href {https://doi.org/10.1103/PhysRevLett.103.220502} {\bibfield  {journal} {\bibinfo  {journal} {Phys. Rev. Lett.}\ }\textbf {\bibinfo {volume} {103}},\ \bibinfo {pages} {220502} (\bibinfo {year} {2009})}\BibitemShut {NoStop}%
\bibitem [{\citenamefont {Chiang}\ and\ \citenamefont {Wocjan}(2010)}]{Chiang2010}%
  \BibitemOpen
  \bibfield  {author} {\bibinfo {author} {\bibfnamefont {C.-F.}\ \bibnamefont {Chiang}}\ and\ \bibinfo {author} {\bibfnamefont {P.}~\bibnamefont {Wocjan}},\ }\bibfield  {title} {\bibinfo {title} {{Quantum algorithm for preparing thermal Gibbs states - detailed analysis}},\ }in\ \href {https://doi.org/10.3233/978-1-60750-547-1-138} {\emph {\bibinfo {booktitle} {Quantum Cryptography and Computing}}},\ \bibinfo {series} {NATO Science for Peace and Security Series - D: Information and Communication Security}, Vol.~\bibinfo {volume} {26}\ (\bibinfo  {publisher} {IOS Press},\ \bibinfo {year} {2010})\ pp.\ \bibinfo {pages} {138--147}\BibitemShut {NoStop}%
\bibitem [{\citenamefont {Chowdhury}\ and\ \citenamefont {Somma}(2017)}]{Chowdhury2017}%
  \BibitemOpen
  \bibfield  {author} {\bibinfo {author} {\bibfnamefont {A.~N.}\ \bibnamefont {Chowdhury}}\ and\ \bibinfo {author} {\bibfnamefont {R.~D.}\ \bibnamefont {Somma}},\ }\bibfield  {title} {\bibinfo {title} {{Quantum algorithms for Gibbs sampling and hitting-time estimation}},\ }\href {https://doi.org/10.26421/QIC17.1-2-3} {\bibfield  {journal} {\bibinfo  {journal} {Quant. Inf. Comp.}\ }\textbf {\bibinfo {volume} {17}},\ \bibinfo {pages} {41} (\bibinfo {year} {2017})}\BibitemShut {NoStop}%
\bibitem [{\citenamefont {van Apeldoorn}\ \emph {et~al.}(2020)\citenamefont {van Apeldoorn}, \citenamefont {Gily^^c3^^a9n}, \citenamefont {Gribling},\ and\ \citenamefont {de~Wolf}}]{vanApeldoorn2020}%
  \BibitemOpen
  \bibfield  {author} {\bibinfo {author} {\bibfnamefont {J.}~\bibnamefont {van Apeldoorn}}, \bibinfo {author} {\bibfnamefont {A.}~\bibnamefont {Gily^^c3^^a9n}}, \bibinfo {author} {\bibfnamefont {S.}~\bibnamefont {Gribling}},\ and\ \bibinfo {author} {\bibfnamefont {R.}~\bibnamefont {de~Wolf}},\ }\bibfield  {title} {\bibinfo {title} {{Quantum SDP-Solvers: Better upper and lower bounds}},\ }\href {https://doi.org/10.22331/q-2020-02-14-230} {\bibfield  {journal} {\bibinfo  {journal} {Quantum}\ }\textbf {\bibinfo {volume} {4}},\ \bibinfo {pages} {230} (\bibinfo {year} {2020})}\BibitemShut {NoStop}%
\bibitem [{\citenamefont {van Apeldoorn}\ and\ \citenamefont {Gily^^c3^^a9n}(2019)}]{vanApeldoorn2019}%
  \BibitemOpen
  \bibfield  {author} {\bibinfo {author} {\bibfnamefont {J.}~\bibnamefont {van Apeldoorn}}\ and\ \bibinfo {author} {\bibfnamefont {A.}~\bibnamefont {Gily^^c3^^a9n}},\ }\bibfield  {title} {\bibinfo {title} {{Improvements in Quantum SDP-Solving with Applications}},\ }in\ \href {https://doi.org/10.4230/LIPIcs.ICALP.2019.99} {\emph {\bibinfo {booktitle} {46th International Colloquium on Automata, Languages, and Programming}}},\ \bibinfo {series} {Leibniz International Proceedings in Informatics}, Vol.\ \bibinfo {volume} {132},\ \bibinfo {editor} {edited by\ \bibinfo {editor} {\bibfnamefont {F.}~\bibnamefont {Meunier}}}\ (\bibinfo  {publisher} {Schloss Dagstuhl -- Leibniz-Zentrum f^^c3^^bcr Informatik},\ \bibinfo {address} {Dagstuhl, Germany},\ \bibinfo {year} {2019})\ pp.\ \bibinfo {pages} {99:1--99:15}\BibitemShut {NoStop}%
\bibitem [{\citenamefont {Gily^^c3^^a9n}\ \emph {et~al.}(2019)\citenamefont {Gily^^c3^^a9n}, \citenamefont {Su}, \citenamefont {Low},\ and\ \citenamefont {Wiebe}}]{Gilyen2019}%
  \BibitemOpen
  \bibfield  {author} {\bibinfo {author} {\bibfnamefont {A.}~\bibnamefont {Gily^^c3^^a9n}}, \bibinfo {author} {\bibfnamefont {Y.}~\bibnamefont {Su}}, \bibinfo {author} {\bibfnamefont {G.~H.}\ \bibnamefont {Low}},\ and\ \bibinfo {author} {\bibfnamefont {N.}~\bibnamefont {Wiebe}},\ }\bibfield  {title} {\bibinfo {title} {{Quantum singular value transformation and beyond: exponential improvements for quantum matrix arithmetics}},\ }in\ \href {https://doi.org/10.1145/3313276.3316366} {\emph {\bibinfo {booktitle} {Proceedings of the 51st Annual ACM SIGACT Symposium on Theory of Computing}}},\ \bibinfo {series and number} {STOC 2019}\ (\bibinfo  {publisher} {Association for Computing Machinery},\ \bibinfo {address} {New York},\ \bibinfo {year} {2019})\ pp.\ \bibinfo {pages} {193--204}\BibitemShut {NoStop}%
\bibitem [{\citenamefont {Gibbs}(1902)}]{Gibbs1902}%
  \BibitemOpen
  \bibfield  {author} {\bibinfo {author} {\bibfnamefont {J.~W.}\ \bibnamefont {Gibbs}},\ }\href@noop {} {\emph {\bibinfo {title} {{Elementary Principles in Statistical Mechanics, Developed with Especial Reference to the Rational Foundation of Thermodynamics}}}}\ (\bibinfo  {publisher} {Charles Scribner's Sons},\ \bibinfo {address} {New York},\ \bibinfo {year} {1902})\ p.\ \bibinfo {pages} {207},\ \bibinfo {note} {first edition}\BibitemShut {NoStop}%
\bibitem [{\citenamefont {Landau}\ and\ \citenamefont {Lifshitz}(1980)}]{Landau1980}%
  \BibitemOpen
  \bibfield  {author} {\bibinfo {author} {\bibfnamefont {L.~D.}\ \bibnamefont {Landau}}\ and\ \bibinfo {author} {\bibfnamefont {E.~M.}\ \bibnamefont {Lifshitz}},\ }\href@noop {} {\emph {\bibinfo {title} {{Statistical Physics, Part 1}}}},\ \bibinfo {edition} {3rd}\ ed.,\ \bibinfo {series} {Course of Theoretical Physics}, Vol.~\bibinfo {volume} {5}\ (\bibinfo  {publisher} {Butterworth-Heinemann},\ \bibinfo {address} {Oxford},\ \bibinfo {year} {1980})\BibitemShut {NoStop}%
\bibitem [{\citenamefont {Toda}\ \emph {et~al.}(1983)\citenamefont {Toda}, \citenamefont {Kubo},\ and\ \citenamefont {Sait^^c5^^8d}}]{Toda1983}%
  \BibitemOpen
  \bibfield  {author} {\bibinfo {author} {\bibfnamefont {M.}~\bibnamefont {Toda}}, \bibinfo {author} {\bibfnamefont {R.}~\bibnamefont {Kubo}},\ and\ \bibinfo {author} {\bibfnamefont {N.}~\bibnamefont {Sait^^c5^^8d}},\ }\href {https://doi.org/10.1007/978-3-642-96698-9} {\emph {\bibinfo {title} {{Statistical Physics I: Equilibrium Statistical Mechanics}}}},\ \bibinfo {series} {Springer Series in Solid-State Sciences}, Vol.~\bibinfo {volume} {30}\ (\bibinfo  {publisher} {Springer-Verlag},\ \bibinfo {address} {Berlin, Heidelberg},\ \bibinfo {year} {1983})\BibitemShut {NoStop}%
\bibitem [{\citenamefont {Callen}(1985)}]{Callen1985}%
  \BibitemOpen
  \bibfield  {author} {\bibinfo {author} {\bibfnamefont {H.~B.}\ \bibnamefont {Callen}},\ }\href@noop {} {\emph {\bibinfo {title} {{Thermodynamics and an Introduction to Thermostatistics}}}},\ \bibinfo {edition} {2nd}\ ed.\ (\bibinfo  {publisher} {John Wiley {\&} Sons},\ \bibinfo {address} {New York},\ \bibinfo {year} {1985})\BibitemShut {NoStop}%
\bibitem [{\citenamefont {Hetherington}(1987)}]{Hetherington1987}%
  \BibitemOpen
  \bibfield  {author} {\bibinfo {author} {\bibfnamefont {J.~H.}\ \bibnamefont {Hetherington}},\ }\bibfield  {title} {\bibinfo {title} {{Solid ${}^3$He magnetism in the classical approximation}},\ }\href {https://doi.org/10.1007/BF00681817} {\bibfield  {journal} {\bibinfo  {journal} {J. Low Temp. Phys.}\ }\textbf {\bibinfo {volume} {66}},\ \bibinfo {pages} {145} (\bibinfo {year} {1987})}\BibitemShut {NoStop}%
\bibitem [{\citenamefont {Challa}\ and\ \citenamefont {Hetherington}(1988{\natexlab{a}})}]{Challa1988_PRL}%
  \BibitemOpen
  \bibfield  {author} {\bibinfo {author} {\bibfnamefont {M.~S.~S.}\ \bibnamefont {Challa}}\ and\ \bibinfo {author} {\bibfnamefont {J.~H.}\ \bibnamefont {Hetherington}},\ }\bibfield  {title} {\bibinfo {title} {{Gaussian Ensemble as an Interpolating Ensemble}},\ }\href {https://doi.org/10.1103/PhysRevLett.60.77} {\bibfield  {journal} {\bibinfo  {journal} {Phys. Rev. Lett.}\ }\textbf {\bibinfo {volume} {60}},\ \bibinfo {pages} {77} (\bibinfo {year} {1988}{\natexlab{a}})}\BibitemShut {NoStop}%
\bibitem [{\citenamefont {Challa}\ and\ \citenamefont {Hetherington}(1988{\natexlab{b}})}]{Challa1988_PRA}%
  \BibitemOpen
  \bibfield  {author} {\bibinfo {author} {\bibfnamefont {M.~S.~S.}\ \bibnamefont {Challa}}\ and\ \bibinfo {author} {\bibfnamefont {J.~H.}\ \bibnamefont {Hetherington}},\ }\bibfield  {title} {\bibinfo {title} {{Gaussian ensemble: An alternate Monte Carlo scheme}},\ }\href {https://doi.org/10.1103/PhysRevA.38.6324} {\bibfield  {journal} {\bibinfo  {journal} {Phys. Rev. A}\ }\textbf {\bibinfo {volume} {38}},\ \bibinfo {pages} {6324} (\bibinfo {year} {1988}{\natexlab{b}})}\BibitemShut {NoStop}%
\bibitem [{\citenamefont {Johal}\ \emph {et~al.}(2003)\citenamefont {Johal}, \citenamefont {Planes},\ and\ \citenamefont {Vives}}]{Johal2003}%
  \BibitemOpen
  \bibfield  {author} {\bibinfo {author} {\bibfnamefont {R.~S.}\ \bibnamefont {Johal}}, \bibinfo {author} {\bibfnamefont {A.}~\bibnamefont {Planes}},\ and\ \bibinfo {author} {\bibfnamefont {E.}~\bibnamefont {Vives}},\ }\bibfield  {title} {\bibinfo {title} {{Statistical mechanics in the extended Gaussian ensemble}},\ }\href {https://doi.org/10.1103/PhysRevE.68.056113} {\bibfield  {journal} {\bibinfo  {journal} {Phys. Rev. E}\ }\textbf {\bibinfo {volume} {68}},\ \bibinfo {pages} {056113} (\bibinfo {year} {2003})}\BibitemShut {NoStop}%
\bibitem [{\citenamefont {Ray}(1991)}]{Ray1991}%
  \BibitemOpen
  \bibfield  {author} {\bibinfo {author} {\bibfnamefont {J.~R.}\ \bibnamefont {Ray}},\ }\bibfield  {title} {\bibinfo {title} {{Microcanonical ensemble Monte Carlo method}},\ }\href {https://doi.org/10.1103/PhysRevA.44.4061} {\bibfield  {journal} {\bibinfo  {journal} {Phys. Rev. A}\ }\textbf {\bibinfo {volume} {44}},\ \bibinfo {pages} {4061} (\bibinfo {year} {1991})}\BibitemShut {NoStop}%
\bibitem [{\citenamefont {Gerling}\ and\ \citenamefont {H{\"{u}}ller}(1993)}]{Gerling1993}%
  \BibitemOpen
  \bibfield  {author} {\bibinfo {author} {\bibfnamefont {R.~W.}\ \bibnamefont {Gerling}}\ and\ \bibinfo {author} {\bibfnamefont {A.}~\bibnamefont {H{\"{u}}ller}},\ }\bibfield  {title} {\bibinfo {title} {{First order phase transitions studied in the dynamical ensemble}},\ }\href {https://doi.org/10.1007/BF02198157} {\bibfield  {journal} {\bibinfo  {journal} {Z. Phys. B}\ }\textbf {\bibinfo {volume} {90}},\ \bibinfo {pages} {207} (\bibinfo {year} {1993})}\BibitemShut {NoStop}%
\bibitem [{\citenamefont {Tsallis}(1988)}]{Tsallis1988}%
  \BibitemOpen
  \bibfield  {author} {\bibinfo {author} {\bibfnamefont {C.}~\bibnamefont {Tsallis}},\ }\bibfield  {title} {\bibinfo {title} {{Possible generalization of Boltzmann-Gibbs statistics}},\ }\href {https://doi.org/10.1007/BF01016429} {\bibfield  {journal} {\bibinfo  {journal} {J. Stat. Phys.}\ }\textbf {\bibinfo {volume} {52}},\ \bibinfo {pages} {479} (\bibinfo {year} {1988})}\BibitemShut {NoStop}%
\bibitem [{\citenamefont {Beck}\ and\ \citenamefont {Cohen}(2003)}]{Beck2003}%
  \BibitemOpen
  \bibfield  {author} {\bibinfo {author} {\bibfnamefont {C.}~\bibnamefont {Beck}}\ and\ \bibinfo {author} {\bibfnamefont {E.}~\bibnamefont {Cohen}},\ }\bibfield  {title} {\bibinfo {title} {{Superstatistics}},\ }\href {https://doi.org/https://doi.org/10.1016/S0378-4371(03)00019-0} {\bibfield  {journal} {\bibinfo  {journal} {Physica A}\ }\textbf {\bibinfo {volume} {322}},\ \bibinfo {pages} {267} (\bibinfo {year} {2003})}\BibitemShut {NoStop}%
\bibitem [{\citenamefont {Cohen}(2004)}]{Cohen2004}%
  \BibitemOpen
  \bibfield  {author} {\bibinfo {author} {\bibfnamefont {E.}~\bibnamefont {Cohen}},\ }\bibfield  {title} {\bibinfo {title} {{Superstatistics}},\ }\href {https://doi.org/https://doi.org/10.1016/j.physd.2004.01.007} {\bibfield  {journal} {\bibinfo  {journal} {Physica D}\ }\textbf {\bibinfo {volume} {193}},\ \bibinfo {pages} {35} (\bibinfo {year} {2004})}\BibitemShut {NoStop}%
\bibitem [{\citenamefont {Costeniuc}\ \emph {et~al.}(2005)\citenamefont {Costeniuc}, \citenamefont {Ellis}, \citenamefont {Touchette},\ and\ \citenamefont {Turkington}}]{Costeniuc2005}%
  \BibitemOpen
  \bibfield  {author} {\bibinfo {author} {\bibfnamefont {M.}~\bibnamefont {Costeniuc}}, \bibinfo {author} {\bibfnamefont {R.~S.}\ \bibnamefont {Ellis}}, \bibinfo {author} {\bibfnamefont {H.}~\bibnamefont {Touchette}},\ and\ \bibinfo {author} {\bibfnamefont {B.}~\bibnamefont {Turkington}},\ }\bibfield  {title} {\bibinfo {title} {{The Generalized Canonical Ensemble and Its Universal Equivalence with the Microcanonical Ensemble}},\ }\href {https://doi.org/10.1007/s10955-005-4407-0} {\bibfield  {journal} {\bibinfo  {journal} {J. Stat. Phys.}\ }\textbf {\bibinfo {volume} {119}},\ \bibinfo {pages} {1283} (\bibinfo {year} {2005})}\BibitemShut {NoStop}%
\bibitem [{\citenamefont {Costeniuc}\ \emph {et~al.}(2006)\citenamefont {Costeniuc}, \citenamefont {Ellis}, \citenamefont {Touchette},\ and\ \citenamefont {Turkington}}]{Costeniuc2006}%
  \BibitemOpen
  \bibfield  {author} {\bibinfo {author} {\bibfnamefont {M.}~\bibnamefont {Costeniuc}}, \bibinfo {author} {\bibfnamefont {R.~S.}\ \bibnamefont {Ellis}}, \bibinfo {author} {\bibfnamefont {H.}~\bibnamefont {Touchette}},\ and\ \bibinfo {author} {\bibfnamefont {B.}~\bibnamefont {Turkington}},\ }\bibfield  {title} {\bibinfo {title} {{Generalized canonical ensembles and ensemble equivalence}},\ }\href {https://doi.org/10.1103/PhysRevE.73.026105} {\bibfield  {journal} {\bibinfo  {journal} {Phys. Rev. E}\ }\textbf {\bibinfo {volume} {73}},\ \bibinfo {pages} {026105} (\bibinfo {year} {2006})}\BibitemShut {NoStop}%
\bibitem [{\citenamefont {Toral}(2006)}]{Toral2006}%
  \BibitemOpen
  \bibfield  {author} {\bibinfo {author} {\bibfnamefont {R.}~\bibnamefont {Toral}},\ }\bibfield  {title} {\bibinfo {title} {{Ensemble equivalence for non-Boltzmannian distributions}},\ }\href {https://doi.org/https://doi.org/10.1016/j.physa.2006.01.040} {\bibfield  {journal} {\bibinfo  {journal} {Physica A}\ }\textbf {\bibinfo {volume} {365}},\ \bibinfo {pages} {85} (\bibinfo {year} {2006})}\BibitemShut {NoStop}%
\bibitem [{\citenamefont {Penrose}\ and\ \citenamefont {Lebowitz}(1971)}]{Penrose1971}%
  \BibitemOpen
  \bibfield  {author} {\bibinfo {author} {\bibfnamefont {O.}~\bibnamefont {Penrose}}\ and\ \bibinfo {author} {\bibfnamefont {J.~L.}\ \bibnamefont {Lebowitz}},\ }\bibfield  {title} {\bibinfo {title} {{Rigorous treatment of metastable states in the van der Waals-Maxwell theory}},\ }\href {https://doi.org/10.1007/BF01019851} {\bibfield  {journal} {\bibinfo  {journal} {J. Stat. Phys.}\ }\textbf {\bibinfo {volume} {3}},\ \bibinfo {pages} {211} (\bibinfo {year} {1971})}\BibitemShut {NoStop}%
\bibitem [{\citenamefont {Ellis}\ \emph {et~al.}(2000)\citenamefont {Ellis}, \citenamefont {Haven},\ and\ \citenamefont {Turkington}}]{Ellis2000}%
  \BibitemOpen
  \bibfield  {author} {\bibinfo {author} {\bibfnamefont {R.~S.}\ \bibnamefont {Ellis}}, \bibinfo {author} {\bibfnamefont {K.}~\bibnamefont {Haven}},\ and\ \bibinfo {author} {\bibfnamefont {B.}~\bibnamefont {Turkington}},\ }\bibfield  {title} {\bibinfo {title} {{Large Deviation Principles and Complete Equivalence and Nonequivalence Results for Pure and Mixed Ensembles}},\ }\href {https://doi.org/10.1023/A:1026446225804} {\bibfield  {journal} {\bibinfo  {journal} {J. Stat. Phys.}\ }\textbf {\bibinfo {volume} {101}},\ \bibinfo {pages} {999} (\bibinfo {year} {2000})}\BibitemShut {NoStop}%
\bibitem [{\citenamefont {Touchette}(2010)}]{Touchette2010}%
  \BibitemOpen
  \bibfield  {author} {\bibinfo {author} {\bibfnamefont {H.}~\bibnamefont {Touchette}},\ }\bibfield  {title} {\bibinfo {title} {{Methods for calculating nonconcave entropies}},\ }\href {https://doi.org/10.1088/1742-5468/2010/05/P05008} {\bibfield  {journal} {\bibinfo  {journal} {J. Stat. Mech.}\ }\textbf {\bibinfo {volume} {2010}},\ \bibinfo {pages} {P05008} (\bibinfo {year} {2010})}\BibitemShut {NoStop}%
\bibitem [{\citenamefont {Yoneta}\ and\ \citenamefont {Shimizu}(2019)}]{Yoneta2019}%
  \BibitemOpen
  \bibfield  {author} {\bibinfo {author} {\bibfnamefont {Y.}~\bibnamefont {Yoneta}}\ and\ \bibinfo {author} {\bibfnamefont {A.}~\bibnamefont {Shimizu}},\ }\bibfield  {title} {\bibinfo {title} {{Squeezed ensemble for systems with first-order phase transitions}},\ }\href {https://doi.org/10.1103/PhysRevB.99.144105} {\bibfield  {journal} {\bibinfo  {journal} {Phys. Rev. B}\ }\textbf {\bibinfo {volume} {99}},\ \bibinfo {pages} {144105} (\bibinfo {year} {2019})}\BibitemShut {NoStop}%
\bibitem [{\citenamefont {Yoneta}\ and\ \citenamefont {Shimizu}(2023)}]{Yoneta2023}%
  \BibitemOpen
  \bibfield  {author} {\bibinfo {author} {\bibfnamefont {Y.}~\bibnamefont {Yoneta}}\ and\ \bibinfo {author} {\bibfnamefont {A.}~\bibnamefont {Shimizu}},\ }\bibfield  {title} {\bibinfo {title} {{Statistical ensembles for phase coexistence states specified by noncommutative additive observables}},\ }\href {https://doi.org/10.1088/1742-5468/accce8} {\bibfield  {journal} {\bibinfo  {journal} {J. Stat. Mech.}\ }\textbf {\bibinfo {volume} {2023}},\ \bibinfo {pages} {053106} (\bibinfo {year} {2023})}\BibitemShut {NoStop}%
\bibitem [{\citenamefont {Ruelle}(1999)}]{Ruelle1999}%
  \BibitemOpen
  \bibfield  {author} {\bibinfo {author} {\bibfnamefont {D.}~\bibnamefont {Ruelle}},\ }\href {https://doi.org/10.1142/4090} {\emph {\bibinfo {title} {{Statistical Mechanics: Rigorous Results}}}}\ (\bibinfo  {publisher} {World Scientific},\ \bibinfo {year} {1999})\BibitemShut {NoStop}%
\bibitem [{\citenamefont {Martin-L^^c3^^b6f}(1979)}]{MartinLof1979}%
  \BibitemOpen
  \bibfield  {author} {\bibinfo {author} {\bibfnamefont {A.}~\bibnamefont {Martin-L^^c3^^b6f}},\ }\bibfield  {title} {\bibinfo {title} {{The equivalence of ensembles and the Gibbs phase rule for classical lattice systems}},\ }\href {https://doi.org/10.1007/BF01012899} {\bibfield  {journal} {\bibinfo  {journal} {J. Stat. Phys.}\ }\textbf {\bibinfo {volume} {20}},\ \bibinfo {pages} {557} (\bibinfo {year} {1979})}\BibitemShut {NoStop}%
\bibitem [{\citenamefont {Georgii}(1995)}]{Georgii1995}%
  \BibitemOpen
  \bibfield  {author} {\bibinfo {author} {\bibfnamefont {H.-O.}\ \bibnamefont {Georgii}},\ }\bibfield  {title} {\bibinfo {title} {{The equivalence of ensembles for classical systems of particles}},\ }\href {https://doi.org/10.1007/BF02179874} {\bibfield  {journal} {\bibinfo  {journal} {J. Stat. Phys.}\ }\textbf {\bibinfo {volume} {80}},\ \bibinfo {pages} {1341} (\bibinfo {year} {1995})}\BibitemShut {NoStop}%
\bibitem [{\citenamefont {Lima}(1972)}]{Lima1972}%
  \BibitemOpen
  \bibfield  {author} {\bibinfo {author} {\bibfnamefont {R.}~\bibnamefont {Lima}},\ }\bibfield  {title} {\bibinfo {title} {{Equivalence of ensembles in quantum lattice systems: States}},\ }\href {https://doi.org/10.1007/BF01877711} {\bibfield  {journal} {\bibinfo  {journal} {Commun. Math. Phys.}\ }\textbf {\bibinfo {volume} {24}},\ \bibinfo {pages} {180} (\bibinfo {year} {1972})}\BibitemShut {NoStop}%
\bibitem [{\citenamefont {M^^c3^^bcller}\ \emph {et~al.}(2015)\citenamefont {M^^c3^^bcller}, \citenamefont {Adlam}, \citenamefont {Masanes},\ and\ \citenamefont {Wiebe}}]{Mueller2015}%
  \BibitemOpen
  \bibfield  {author} {\bibinfo {author} {\bibfnamefont {M.~P.}\ \bibnamefont {M^^c3^^bcller}}, \bibinfo {author} {\bibfnamefont {E.}~\bibnamefont {Adlam}}, \bibinfo {author} {\bibfnamefont {L.}~\bibnamefont {Masanes}},\ and\ \bibinfo {author} {\bibfnamefont {N.}~\bibnamefont {Wiebe}},\ }\bibfield  {title} {\bibinfo {title} {{Thermalization and Canonical Typicality in Translation-Invariant Quantum Lattice Systems}},\ }\href {https://doi.org/10.1007/s00220-015-2473-y} {\bibfield  {journal} {\bibinfo  {journal} {Commun. Math. Phys.}\ }\textbf {\bibinfo {volume} {340}},\ \bibinfo {pages} {499} (\bibinfo {year} {2015})}\BibitemShut {NoStop}%
\bibitem [{\citenamefont {Brandao}\ and\ \citenamefont {Cramer}(2015)}]{Brandao2015}%
  \BibitemOpen
  \bibfield  {author} {\bibinfo {author} {\bibfnamefont {F.~G. S.~L.}\ \bibnamefont {Brandao}}\ and\ \bibinfo {author} {\bibfnamefont {M.}~\bibnamefont {Cramer}},\ }\href {https://arxiv.org/abs/1502.03263} {\bibinfo {title} {{Equivalence of Statistical Mechanical Ensembles for Non-Critical Quantum Systems}}} (\bibinfo {year} {2015}),\ \Eprint {https://arxiv.org/abs/1502.03263} {arXiv:1502.03263 [quant-ph]} \BibitemShut {NoStop}%
\bibitem [{\citenamefont {Tasaki}(2018)}]{Tasaki2018}%
  \BibitemOpen
  \bibfield  {author} {\bibinfo {author} {\bibfnamefont {H.}~\bibnamefont {Tasaki}},\ }\bibfield  {title} {\bibinfo {title} {{On the Local Equivalence Between the Canonical and the Microcanonical Ensembles for Quantum Spin Systems}},\ }\href {https://doi.org/10.1007/s10955-018-2077-y} {\bibfield  {journal} {\bibinfo  {journal} {J. Stat. Phys.}\ }\textbf {\bibinfo {volume} {172}},\ \bibinfo {pages} {905} (\bibinfo {year} {2018})}\BibitemShut {NoStop}%
\bibitem [{\citenamefont {Kim}\ \emph {et~al.}(2010)\citenamefont {Kim}, \citenamefont {Keyes},\ and\ \citenamefont {Straub}}]{Kim2010}%
  \BibitemOpen
  \bibfield  {author} {\bibinfo {author} {\bibfnamefont {J.}~\bibnamefont {Kim}}, \bibinfo {author} {\bibfnamefont {T.}~\bibnamefont {Keyes}},\ and\ \bibinfo {author} {\bibfnamefont {J.~E.}\ \bibnamefont {Straub}},\ }\bibfield  {title} {\bibinfo {title} {{Generalized Replica Exchange Method}},\ }\href {https://doi.org/10.1063/1.3432176} {\bibfield  {journal} {\bibinfo  {journal} {J. Chem. Phys.}\ }\textbf {\bibinfo {volume} {132}},\ \bibinfo {pages} {224107} (\bibinfo {year} {2010})}\BibitemShut {NoStop}%
\bibitem [{\citenamefont {Schierz}\ \emph {et~al.}(2016)\citenamefont {Schierz}, \citenamefont {Zierenberg},\ and\ \citenamefont {Janke}}]{Schierz2016}%
  \BibitemOpen
  \bibfield  {author} {\bibinfo {author} {\bibfnamefont {P.}~\bibnamefont {Schierz}}, \bibinfo {author} {\bibfnamefont {J.}~\bibnamefont {Zierenberg}},\ and\ \bibinfo {author} {\bibfnamefont {W.}~\bibnamefont {Janke}},\ }\bibfield  {title} {\bibinfo {title} {{First-order phase transitions in the real microcanonical ensemble}},\ }\href {https://doi.org/10.1103/PhysRevE.94.021301} {\bibfield  {journal} {\bibinfo  {journal} {Phys. Rev. E}\ }\textbf {\bibinfo {volume} {94}},\ \bibinfo {pages} {021301} (\bibinfo {year} {2016})}\BibitemShut {NoStop}%
\bibitem [{\citenamefont {Sugiura}\ and\ \citenamefont {Shimizu}(2012)}]{Sugiura2012}%
  \BibitemOpen
  \bibfield  {author} {\bibinfo {author} {\bibfnamefont {S.}~\bibnamefont {Sugiura}}\ and\ \bibinfo {author} {\bibfnamefont {A.}~\bibnamefont {Shimizu}},\ }\bibfield  {title} {\bibinfo {title} {{Thermal Pure Quantum States at Finite Temperature}},\ }\href {https://doi.org/10.1103/PhysRevLett.108.240401} {\bibfield  {journal} {\bibinfo  {journal} {Phys. Rev. Lett.}\ }\textbf {\bibinfo {volume} {108}},\ \bibinfo {pages} {240401} (\bibinfo {year} {2012})}\BibitemShut {NoStop}%
\bibitem [{\citenamefont {Temme}\ \emph {et~al.}(2011)\citenamefont {Temme}, \citenamefont {Osborne}, \citenamefont {Vollbrecht}, \citenamefont {Poulin},\ and\ \citenamefont {Verstraete}}]{Temme2011}%
  \BibitemOpen
  \bibfield  {author} {\bibinfo {author} {\bibfnamefont {K.}~\bibnamefont {Temme}}, \bibinfo {author} {\bibfnamefont {T.~J.}\ \bibnamefont {Osborne}}, \bibinfo {author} {\bibfnamefont {K.~G.}\ \bibnamefont {Vollbrecht}}, \bibinfo {author} {\bibfnamefont {D.}~\bibnamefont {Poulin}},\ and\ \bibinfo {author} {\bibfnamefont {F.}~\bibnamefont {Verstraete}},\ }\bibfield  {title} {\bibinfo {title} {{Quantum Metropolis sampling}},\ }\href {https://doi.org/10.1038/nature09770} {\bibfield  {journal} {\bibinfo  {journal} {Nature}\ }\textbf {\bibinfo {volume} {471}},\ \bibinfo {pages} {87} (\bibinfo {year} {2011})}\BibitemShut {NoStop}%
\bibitem [{\citenamefont {Yung}\ and\ \citenamefont {Aspuru-Guzik}(2012)}]{Yung2012}%
  \BibitemOpen
  \bibfield  {author} {\bibinfo {author} {\bibfnamefont {M.-H.}\ \bibnamefont {Yung}}\ and\ \bibinfo {author} {\bibfnamefont {A.}~\bibnamefont {Aspuru-Guzik}},\ }\bibfield  {title} {\bibinfo {title} {{A quantum^^e2^^80^^93quantum Metropolis algorithm}},\ }\href {https://doi.org/10.1073/pnas.1111758109} {\bibfield  {journal} {\bibinfo  {journal} {Proc. Natl. Acad. Sci.}\ }\textbf {\bibinfo {volume} {109}},\ \bibinfo {pages} {754} (\bibinfo {year} {2012})}\BibitemShut {NoStop}%
\bibitem [{\citenamefont {Chen}\ and\ \citenamefont {Brand\~{a}o}(2021)}]{Chen2021}%
  \BibitemOpen
  \bibfield  {author} {\bibinfo {author} {\bibfnamefont {C.-F.}\ \bibnamefont {Chen}}\ and\ \bibinfo {author} {\bibfnamefont {F.~G. S.~L.}\ \bibnamefont {Brand\~{a}o}},\ }\href {https://arxiv.org/abs/2112.07646} {\bibinfo {title} {{Fast Thermalization from the Eigenstate Thermalization Hypothesis}}} (\bibinfo {year} {2021}),\ \Eprint {https://arxiv.org/abs/2112.07646} {arXiv:2112.07646 [quant-ph]} \BibitemShut {NoStop}%
\bibitem [{\citenamefont {Wocjan}\ and\ \citenamefont {Temme}(2023)}]{Wocjan2023}%
  \BibitemOpen
  \bibfield  {author} {\bibinfo {author} {\bibfnamefont {P.}~\bibnamefont {Wocjan}}\ and\ \bibinfo {author} {\bibfnamefont {K.}~\bibnamefont {Temme}},\ }\bibfield  {title} {\bibinfo {title} {{Szegedy Walk Unitaries for Quantum Maps}},\ }\href {https://doi.org/10.1007/s00220-023-04797-4} {\bibfield  {journal} {\bibinfo  {journal} {Commun. Math. Phys.}\ }\textbf {\bibinfo {volume} {402}},\ \bibinfo {pages} {3201} (\bibinfo {year} {2023})}\BibitemShut {NoStop}%
\bibitem [{\citenamefont {Rall}\ \emph {et~al.}(2023)\citenamefont {Rall}, \citenamefont {Wang},\ and\ \citenamefont {Wocjan}}]{Rall2023_2}%
  \BibitemOpen
  \bibfield  {author} {\bibinfo {author} {\bibfnamefont {P.}~\bibnamefont {Rall}}, \bibinfo {author} {\bibfnamefont {C.}~\bibnamefont {Wang}},\ and\ \bibinfo {author} {\bibfnamefont {P.}~\bibnamefont {Wocjan}},\ }\bibfield  {title} {\bibinfo {title} {{Thermal State Preparation via Rounding Promises}},\ }\href {https://doi.org/10.22331/q-2023-10-10-1132} {\bibfield  {journal} {\bibinfo  {journal} {Quantum}\ }\textbf {\bibinfo {volume} {7}},\ \bibinfo {pages} {1132} (\bibinfo {year} {2023})}\BibitemShut {NoStop}%
\bibitem [{\citenamefont {Chen}\ \emph {et~al.}(2023)\citenamefont {Chen}, \citenamefont {Kastoryano}, \citenamefont {Brand^^c3^^a3o},\ and\ \citenamefont {Gily^^c3^^a9n}}]{Chen2023}%
  \BibitemOpen
  \bibfield  {author} {\bibinfo {author} {\bibfnamefont {C.-F.}\ \bibnamefont {Chen}}, \bibinfo {author} {\bibfnamefont {M.~J.}\ \bibnamefont {Kastoryano}}, \bibinfo {author} {\bibfnamefont {F.~G. S.~L.}\ \bibnamefont {Brand^^c3^^a3o}},\ and\ \bibinfo {author} {\bibfnamefont {A.}~\bibnamefont {Gily^^c3^^a9n}},\ }\href {https://arxiv.org/abs/2303.18224} {\bibinfo {title} {{Quantum Thermal State Preparation}}} (\bibinfo {year} {2023}),\ \Eprint {https://arxiv.org/abs/2303.18224} {arXiv:2303.18224 [quant-ph]} \BibitemShut {NoStop}%
\bibitem [{\citenamefont {Bilgin}\ and\ \citenamefont {Boixo}(2010)}]{Bilgin2010}%
  \BibitemOpen
  \bibfield  {author} {\bibinfo {author} {\bibfnamefont {E.}~\bibnamefont {Bilgin}}\ and\ \bibinfo {author} {\bibfnamefont {S.}~\bibnamefont {Boixo}},\ }\bibfield  {title} {\bibinfo {title} {{Preparing Thermal States of Quantum Systems by Dimension Reduction}},\ }\href {https://doi.org/10.1103/PhysRevLett.105.170405} {\bibfield  {journal} {\bibinfo  {journal} {Phys. Rev. Lett.}\ }\textbf {\bibinfo {volume} {105}},\ \bibinfo {pages} {170405} (\bibinfo {year} {2010})}\BibitemShut {NoStop}%
\bibitem [{\citenamefont {Seki}\ and\ \citenamefont {Yunoki}(2022)}]{Seki2022}%
  \BibitemOpen
  \bibfield  {author} {\bibinfo {author} {\bibfnamefont {K.}~\bibnamefont {Seki}}\ and\ \bibinfo {author} {\bibfnamefont {S.}~\bibnamefont {Yunoki}},\ }\bibfield  {title} {\bibinfo {title} {{Energy-filtered random-phase states as microcanonical thermal pure quantum states}},\ }\href {https://doi.org/10.1103/PhysRevB.106.155111} {\bibfield  {journal} {\bibinfo  {journal} {Phys. Rev. B}\ }\textbf {\bibinfo {volume} {106}},\ \bibinfo {pages} {155111} (\bibinfo {year} {2022})}\BibitemShut {NoStop}%
\bibitem [{\citenamefont {Mizukami}\ and\ \citenamefont {Koga}(2023)}]{Mizukami2023}%
  \BibitemOpen
  \bibfield  {author} {\bibinfo {author} {\bibfnamefont {K.}~\bibnamefont {Mizukami}}\ and\ \bibinfo {author} {\bibfnamefont {A.}~\bibnamefont {Koga}},\ }\bibfield  {title} {\bibinfo {title} {{Quantum algorithm for the microcanonical thermal pure quantum state method}},\ }\href {https://doi.org/10.1103/PhysRevA.108.012404} {\bibfield  {journal} {\bibinfo  {journal} {Phys. Rev. A}\ }\textbf {\bibinfo {volume} {108}},\ \bibinfo {pages} {012404} (\bibinfo {year} {2023})}\BibitemShut {NoStop}%
\bibitem [{\citenamefont {Martyn}\ \emph {et~al.}(2021)\citenamefont {Martyn}, \citenamefont {Rossi}, \citenamefont {Tan},\ and\ \citenamefont {Chuang}}]{Martyn2021}%
  \BibitemOpen
  \bibfield  {author} {\bibinfo {author} {\bibfnamefont {J.~M.}\ \bibnamefont {Martyn}}, \bibinfo {author} {\bibfnamefont {Z.~M.}\ \bibnamefont {Rossi}}, \bibinfo {author} {\bibfnamefont {A.~K.}\ \bibnamefont {Tan}},\ and\ \bibinfo {author} {\bibfnamefont {I.~L.}\ \bibnamefont {Chuang}},\ }\bibfield  {title} {\bibinfo {title} {{Grand Unification of Quantum Algorithms}},\ }\href {https://doi.org/10.1103/PRXQuantum.2.040203} {\bibfield  {journal} {\bibinfo  {journal} {PRX Quantum}\ }\textbf {\bibinfo {volume} {2}},\ \bibinfo {pages} {040203} (\bibinfo {year} {2021})}\BibitemShut {NoStop}%
\bibitem [{Note1()}]{Note1}%
  \BibitemOpen
  \bibinfo {note} {The assumption of the concavity of the limiting function is made to simplify the discussion. When the limiting function is nonconcave, several difficulties arise, such as the thermodynamic entropy becoming ill-defined and the canonical ensemble giving a statistical mixture of macroscopically distinct states~\cite {Gross2001,Yoneta2019}. However, even in such cases, all the results of this paper remain valid, provided that $s$ in the entropy formula~\protect \eqref {eq:entropy-formula_canonical} for the canonical ensemble, as well as all instances of $s$ appearing in the asymptotic analysis of the computational cost for the thermal state preparation algorithm applied to the canonical ensemble, are replaced with the concave envelope of the limiting function.}\BibitemShut {Stop}%
\bibitem [{\citenamefont {Low}\ and\ \citenamefont {Chuang}(2017{\natexlab{a}})}]{Low2017_2}%
  \BibitemOpen
  \bibfield  {author} {\bibinfo {author} {\bibfnamefont {G.~H.}\ \bibnamefont {Low}}\ and\ \bibinfo {author} {\bibfnamefont {I.~L.}\ \bibnamefont {Chuang}},\ }\bibfield  {title} {\bibinfo {title} {{Optimal Hamiltonian Simulation by Quantum Signal Processing}},\ }\href {https://doi.org/10.1103/PhysRevLett.118.010501} {\bibfield  {journal} {\bibinfo  {journal} {Phys. Rev. Lett.}\ }\textbf {\bibinfo {volume} {118}},\ \bibinfo {pages} {010501} (\bibinfo {year} {2017}{\natexlab{a}})}\BibitemShut {NoStop}%
\bibitem [{\citenamefont {Low}\ and\ \citenamefont {Chuang}(2019)}]{Low2019}%
  \BibitemOpen
  \bibfield  {author} {\bibinfo {author} {\bibfnamefont {G.~H.}\ \bibnamefont {Low}}\ and\ \bibinfo {author} {\bibfnamefont {I.~L.}\ \bibnamefont {Chuang}},\ }\bibfield  {title} {\bibinfo {title} {{Hamiltonian Simulation by Qubitization}},\ }\href {https://doi.org/10.22331/q-2019-07-12-163} {\bibfield  {journal} {\bibinfo  {journal} {{Quantum}}\ }\textbf {\bibinfo {volume} {3}},\ \bibinfo {pages} {163} (\bibinfo {year} {2019})}\BibitemShut {NoStop}%
\bibitem [{\citenamefont {Rall}\ and\ \citenamefont {Fuller}(2023)}]{Rall2023_1}%
  \BibitemOpen
  \bibfield  {author} {\bibinfo {author} {\bibfnamefont {P.}~\bibnamefont {Rall}}\ and\ \bibinfo {author} {\bibfnamefont {B.}~\bibnamefont {Fuller}},\ }\bibfield  {title} {\bibinfo {title} {{Amplitude Estimation from Quantum Signal Processing}},\ }\href {https://doi.org/10.22331/q-2023-03-02-937} {\bibfield  {journal} {\bibinfo  {journal} {{Quantum}}\ }\textbf {\bibinfo {volume} {7}},\ \bibinfo {pages} {937} (\bibinfo {year} {2023})}\BibitemShut {NoStop}%
\bibitem [{\citenamefont {Rall}(2021)}]{Rall2021}%
  \BibitemOpen
  \bibfield  {author} {\bibinfo {author} {\bibfnamefont {P.}~\bibnamefont {Rall}},\ }\bibfield  {title} {\bibinfo {title} {{Faster Coherent Quantum Algorithms for Phase, Energy, and Amplitude Estimation}},\ }\href {https://doi.org/10.22331/q-2021-10-19-566} {\bibfield  {journal} {\bibinfo  {journal} {{Quantum}}\ }\textbf {\bibinfo {volume} {5}},\ \bibinfo {pages} {566} (\bibinfo {year} {2021})}\BibitemShut {NoStop}%
\bibitem [{\citenamefont {Low}\ \emph {et~al.}(2016)\citenamefont {Low}, \citenamefont {Yoder},\ and\ \citenamefont {Chuang}}]{Low2016}%
  \BibitemOpen
  \bibfield  {author} {\bibinfo {author} {\bibfnamefont {G.~H.}\ \bibnamefont {Low}}, \bibinfo {author} {\bibfnamefont {T.~J.}\ \bibnamefont {Yoder}},\ and\ \bibinfo {author} {\bibfnamefont {I.~L.}\ \bibnamefont {Chuang}},\ }\bibfield  {title} {\bibinfo {title} {{Methodology of Resonant Equiangular Composite Quantum Gates}},\ }\href {https://doi.org/10.1103/PhysRevX.6.041067} {\bibfield  {journal} {\bibinfo  {journal} {Phys. Rev. X}\ }\textbf {\bibinfo {volume} {6}},\ \bibinfo {pages} {041067} (\bibinfo {year} {2016})}\BibitemShut {NoStop}%
\bibitem [{\citenamefont {Haah}(2019)}]{Haah2019}%
  \BibitemOpen
  \bibfield  {author} {\bibinfo {author} {\bibfnamefont {J.}~\bibnamefont {Haah}},\ }\bibfield  {title} {\bibinfo {title} {{Product Decomposition of Periodic Functions in Quantum Signal Processing}},\ }\href {https://doi.org/10.22331/q-2019-10-07-190} {\bibfield  {journal} {\bibinfo  {journal} {Quantum}\ }\textbf {\bibinfo {volume} {3}},\ \bibinfo {pages} {190} (\bibinfo {year} {2019})}\BibitemShut {NoStop}%
\bibitem [{\citenamefont {Chao}\ \emph {et~al.}(2020)\citenamefont {Chao}, \citenamefont {Ding}, \citenamefont {Gily^^c3^^a9n}, \citenamefont {Huang},\ and\ \citenamefont {Szegedy}}]{Chao2020}%
  \BibitemOpen
  \bibfield  {author} {\bibinfo {author} {\bibfnamefont {R.}~\bibnamefont {Chao}}, \bibinfo {author} {\bibfnamefont {D.}~\bibnamefont {Ding}}, \bibinfo {author} {\bibfnamefont {A.}~\bibnamefont {Gily^^c3^^a9n}}, \bibinfo {author} {\bibfnamefont {C.}~\bibnamefont {Huang}},\ and\ \bibinfo {author} {\bibfnamefont {M.}~\bibnamefont {Szegedy}},\ }\href {https://arxiv.org/abs/2003.02831} {\bibinfo {title} {{Finding Angles for Quantum Signal Processing with Machine Precision}}} (\bibinfo {year} {2020}),\ \Eprint {https://arxiv.org/abs/2003.02831} {arXiv:2003.02831 [quant-ph]} \BibitemShut {NoStop}%
\bibitem [{\citenamefont {Dong}\ \emph {et~al.}(2021)\citenamefont {Dong}, \citenamefont {Meng}, \citenamefont {Whaley},\ and\ \citenamefont {Lin}}]{Dong2021}%
  \BibitemOpen
  \bibfield  {author} {\bibinfo {author} {\bibfnamefont {Y.}~\bibnamefont {Dong}}, \bibinfo {author} {\bibfnamefont {X.}~\bibnamefont {Meng}}, \bibinfo {author} {\bibfnamefont {K.~B.}\ \bibnamefont {Whaley}},\ and\ \bibinfo {author} {\bibfnamefont {L.}~\bibnamefont {Lin}},\ }\bibfield  {title} {\bibinfo {title} {{Efficient phase-factor evaluation in quantum signal processing}},\ }\href {https://doi.org/10.1103/PhysRevA.103.042419} {\bibfield  {journal} {\bibinfo  {journal} {Phys. Rev. A}\ }\textbf {\bibinfo {volume} {103}},\ \bibinfo {pages} {042419} (\bibinfo {year} {2021})}\BibitemShut {NoStop}%
\bibitem [{\citenamefont {Ying}(2022)}]{Ying2022}%
  \BibitemOpen
  \bibfield  {author} {\bibinfo {author} {\bibfnamefont {L.}~\bibnamefont {Ying}},\ }\bibfield  {title} {\bibinfo {title} {{Stable Factorization for Phase Factors of Quantum Signal Processing}},\ }\href {https://doi.org/10.22331/q-2022-10-20-842} {\bibfield  {journal} {\bibinfo  {journal} {{Quantum}}\ }\textbf {\bibinfo {volume} {6}},\ \bibinfo {pages} {842} (\bibinfo {year} {2022})}\BibitemShut {NoStop}%
\bibitem [{\citenamefont {Berntson}\ and\ \citenamefont {S^^c3^^bcnderhauf}(2024)}]{Berntson2024}%
  \BibitemOpen
  \bibfield  {author} {\bibinfo {author} {\bibfnamefont {B.~K.}\ \bibnamefont {Berntson}}\ and\ \bibinfo {author} {\bibfnamefont {C.}~\bibnamefont {S^^c3^^bcnderhauf}},\ }\href {https://doi.org/10.48550/arXiv.2406.04246} {\bibinfo {title} {{Complementary Polynomials in Quantum Signal Processing}}} (\bibinfo {year} {2024}),\ \Eprint {https://arxiv.org/abs/2406.04246} {arXiv:2406.04246 [quant-ph]} \BibitemShut {NoStop}%
\bibitem [{\citenamefont {Alexis}\ \emph {et~al.}(2024)\citenamefont {Alexis}, \citenamefont {Lin}, \citenamefont {Mnatsakanyan}, \citenamefont {Thiele},\ and\ \citenamefont {Wang}}]{Alexis2024}%
  \BibitemOpen
  \bibfield  {author} {\bibinfo {author} {\bibfnamefont {M.}~\bibnamefont {Alexis}}, \bibinfo {author} {\bibfnamefont {L.}~\bibnamefont {Lin}}, \bibinfo {author} {\bibfnamefont {G.}~\bibnamefont {Mnatsakanyan}}, \bibinfo {author} {\bibfnamefont {C.}~\bibnamefont {Thiele}},\ and\ \bibinfo {author} {\bibfnamefont {J.}~\bibnamefont {Wang}},\ }\href {https://doi.org/10.48550/arXiv.2407.05634} {\bibinfo {title} {{Infinite Quantum Signal Processing for Arbitrary Szeg\H{o} Functions}}} (\bibinfo {year} {2024}),\ \Eprint {https://arxiv.org/abs/2407.05634} {arXiv:2407.05634 [quant-ph]} \BibitemShut {NoStop}%
\bibitem [{\citenamefont {Garratt}\ and\ \citenamefont {Choi}(2024)}]{Garratt2024}%
  \BibitemOpen
  \bibfield  {author} {\bibinfo {author} {\bibfnamefont {S.~J.}\ \bibnamefont {Garratt}}\ and\ \bibinfo {author} {\bibfnamefont {S.}~\bibnamefont {Choi}},\ }\href {https://doi.org/10.48550/arXiv.2407.07893} {\bibinfo {title} {{Quantum Algorithm to Prepare Quasi-Stationary States}}} (\bibinfo {year} {2024}),\ \Eprint {https://arxiv.org/abs/2407.07893} {arXiv:2407.07893 [quant-ph]} \BibitemShut {NoStop}%
\bibitem [{\citenamefont {Yoder}\ \emph {et~al.}(2014)\citenamefont {Yoder}, \citenamefont {Low},\ and\ \citenamefont {Chuang}}]{Yoder2014}%
  \BibitemOpen
  \bibfield  {author} {\bibinfo {author} {\bibfnamefont {T.~J.}\ \bibnamefont {Yoder}}, \bibinfo {author} {\bibfnamefont {G.~H.}\ \bibnamefont {Low}},\ and\ \bibinfo {author} {\bibfnamefont {I.~L.}\ \bibnamefont {Chuang}},\ }\bibfield  {title} {\bibinfo {title} {{Fixed-Point Quantum Search with an Optimal Number of Queries}},\ }\href {https://doi.org/10.1103/PhysRevLett.113.210501} {\bibfield  {journal} {\bibinfo  {journal} {Phys. Rev. Lett.}\ }\textbf {\bibinfo {volume} {113}},\ \bibinfo {pages} {210501} (\bibinfo {year} {2014})}\BibitemShut {NoStop}%
\bibitem [{\citenamefont {Sachdeva}\ and\ \citenamefont {Vishnoi}(2014)}]{Sachdeva2014}%
  \BibitemOpen
  \bibfield  {author} {\bibinfo {author} {\bibfnamefont {S.}~\bibnamefont {Sachdeva}}\ and\ \bibinfo {author} {\bibfnamefont {N.~K.}\ \bibnamefont {Vishnoi}},\ }\bibfield  {title} {\bibinfo {title} {{Faster Algorithms via Approximation Theory}},\ }\href {https://doi.org/10.1561/0400000065} {\bibfield  {journal} {\bibinfo  {journal} {Found. Trends Theor. Comput. Sci.}\ }\textbf {\bibinfo {volume} {9}},\ \bibinfo {pages} {125} (\bibinfo {year} {2014})}\BibitemShut {NoStop}%
\bibitem [{\citenamefont {Low}\ and\ \citenamefont {Chuang}(2017{\natexlab{b}})}]{Low2017_1}%
  \BibitemOpen
  \bibfield  {author} {\bibinfo {author} {\bibfnamefont {G.~H.}\ \bibnamefont {Low}}\ and\ \bibinfo {author} {\bibfnamefont {I.~L.}\ \bibnamefont {Chuang}},\ }\href {https://arxiv.org/abs/1707.05391} {\bibinfo {title} {{Hamiltonian Simulation by Uniform Spectral Amplification}}} (\bibinfo {year} {2017}{\natexlab{b}}),\ \Eprint {https://arxiv.org/abs/1707.05391} {arXiv:1707.05391 [quant-ph]} \BibitemShut {NoStop}%
\bibitem [{Note2()}]{Note2}%
  \BibitemOpen
  \bibinfo {note} {Similarly, one can alternatively employ the pure-state method~\cite {Sugiura2012,Sugiura2013,Yoneta2024}. While this approach requires additional computational cost for preparing infinite-temperature states, it allows one to eliminate the ancilla register of $N$ qubits that is required for purification. Even in this case, however, the number of queries to $U_H$ required to construct each individual state remains identical to that in the purification method.}\BibitemShut {Stop}%
\bibitem [{\citenamefont {Hoorfar}(2008)}]{Hoorfar2008}%
  \BibitemOpen
  \bibfield  {author} {\bibinfo {author} {\bibfnamefont {H.~M.}\ \bibnamefont {Hoorfar}, \bibfnamefont {Abdolhossein}},\ }\bibfield  {title} {\bibinfo {title} {{Inequalities on the Lambert $W$ function and hyperpower function}},\ }\href {http://eudml.org/doc/130024} {\bibfield  {journal} {\bibinfo  {journal} {J. Inequal. Pure and Appl. Math.}\ }\textbf {\bibinfo {volume} {9}},\ \bibinfo {pages} {5} (\bibinfo {year} {2008})}\BibitemShut {NoStop}%
\bibitem [{\citenamefont {Grover}(1996)}]{Grover1996}%
  \BibitemOpen
  \bibfield  {author} {\bibinfo {author} {\bibfnamefont {L.~K.}\ \bibnamefont {Grover}},\ }\bibfield  {title} {\bibinfo {title} {{A fast quantum mechanical algorithm for database search}},\ }in\ \href {https://doi.org/10.1145/237814.237866} {\emph {\bibinfo {booktitle} {Proceedings of the 28th annual ACM symposium on Theory of Computing}}},\ \bibinfo {series and number} {STOC '96}\ (\bibinfo  {publisher} {Association for Computing Machinery},\ \bibinfo {address} {New York},\ \bibinfo {year} {1996})\ pp.\ \bibinfo {pages} {212--219}\BibitemShut {NoStop}%
\bibitem [{\citenamefont {Grover}(2005)}]{Grover2005}%
  \BibitemOpen
  \bibfield  {author} {\bibinfo {author} {\bibfnamefont {L.~K.}\ \bibnamefont {Grover}},\ }\bibfield  {title} {\bibinfo {title} {{Fixed-Point Quantum Search}},\ }\href {https://doi.org/10.1103/PhysRevLett.95.150501} {\bibfield  {journal} {\bibinfo  {journal} {Phys. Rev. Lett.}\ }\textbf {\bibinfo {volume} {95}},\ \bibinfo {pages} {150501} (\bibinfo {year} {2005})}\BibitemShut {NoStop}%
\bibitem [{\citenamefont {H\o{}yer}(2000)}]{Hoyer2000}%
  \BibitemOpen
  \bibfield  {author} {\bibinfo {author} {\bibfnamefont {P.}~\bibnamefont {H\o{}yer}},\ }\bibfield  {title} {\bibinfo {title} {{Arbitrary phases in quantum amplitude amplification}},\ }\href {https://doi.org/10.1103/PhysRevA.62.052304} {\bibfield  {journal} {\bibinfo  {journal} {Phys. Rev. A}\ }\textbf {\bibinfo {volume} {62}},\ \bibinfo {pages} {052304} (\bibinfo {year} {2000})}\BibitemShut {NoStop}%
\bibitem [{\citenamefont {Long}(2001)}]{Long2001}%
  \BibitemOpen
  \bibfield  {author} {\bibinfo {author} {\bibfnamefont {G.~L.}\ \bibnamefont {Long}},\ }\bibfield  {title} {\bibinfo {title} {{Grover algorithm with zero theoretical failure rate}},\ }\href {https://doi.org/10.1103/PhysRevA.64.022307} {\bibfield  {journal} {\bibinfo  {journal} {Phys. Rev. A}\ }\textbf {\bibinfo {volume} {64}},\ \bibinfo {pages} {022307} (\bibinfo {year} {2001})}\BibitemShut {NoStop}%
\bibitem [{Note3()}]{Note3}%
  \BibitemOpen
  \bibinfo {note} {In particular, in the limit $n\to \infty $, this ensemble approaches the microcanonical ensemble with an energy shell $[\mu -\Delta , \mu +\Delta ]$.}\BibitemShut {Stop}%
\bibitem [{\citenamefont {Gross}(2001)}]{Gross2001}%
  \BibitemOpen
  \bibfield  {author} {\bibinfo {author} {\bibfnamefont {D.~H.~E.}\ \bibnamefont {Gross}},\ }\href {https://doi.org/10.1142/4340} {\emph {\bibinfo {title} {{Microcanonical thermodynamics: Phase transitions in Small systems}}}},\ \bibinfo {series} {Lectures Notes in Physics}, Vol.~\bibinfo {volume} {66}\ (\bibinfo  {publisher} {World Scientific},\ \bibinfo {address} {Singapore},\ \bibinfo {year} {2001})\BibitemShut {NoStop}%
\bibitem [{\citenamefont {Abramowitz}\ and\ \citenamefont {Stegun}(1972)}]{Abramowitz1964}%
  \BibitemOpen
  \bibinfo {editor} {\bibfnamefont {M.}~\bibnamefont {Abramowitz}}\ and\ \bibinfo {editor} {\bibfnamefont {I.~A.}\ \bibnamefont {Stegun}},\ eds.,\ \href@noop {} {\emph {\bibinfo {title} {{Handbook of Mathematical Functions with Formulas, Graphs, and Mathematical Tables}}}},\ \bibinfo {edition} {10th}\ ed.,\ \bibinfo {series} {Applied Mathematics Series}, Vol.~\bibinfo {volume} {55}\ (\bibinfo  {publisher} {National Bureau of Standards},\ \bibinfo {address} {Washington, D.C.},\ \bibinfo {year} {1972})\BibitemShut {NoStop}%
\bibitem [{\citenamefont {Sugiura}\ and\ \citenamefont {Shimizu}(2013)}]{Sugiura2013}%
  \BibitemOpen
  \bibfield  {author} {\bibinfo {author} {\bibfnamefont {S.}~\bibnamefont {Sugiura}}\ and\ \bibinfo {author} {\bibfnamefont {A.}~\bibnamefont {Shimizu}},\ }\bibfield  {title} {\bibinfo {title} {{Canonical Thermal Pure Quantum State}},\ }\href {https://doi.org/10.1103/PhysRevLett.111.010401} {\bibfield  {journal} {\bibinfo  {journal} {Phys. Rev. Lett.}\ }\textbf {\bibinfo {volume} {111}},\ \bibinfo {pages} {010401} (\bibinfo {year} {2013})}\BibitemShut {NoStop}%
\bibitem [{\citenamefont {Yoneta}(2024)}]{Yoneta2024}%
  \BibitemOpen
  \bibfield  {author} {\bibinfo {author} {\bibfnamefont {Y.}~\bibnamefont {Yoneta}},\ }\bibfield  {title} {\bibinfo {title} {{Thermal pure states for systems with antiunitary symmetries and their tensor network representations}},\ }\href {https://doi.org/10.1103/PhysRevResearch.6.L042062} {\bibfield  {journal} {\bibinfo  {journal} {Phys. Rev. Res.}\ }\textbf {\bibinfo {volume} {6}},\ \bibinfo {pages} {L042062} (\bibinfo {year} {2024})}\BibitemShut {NoStop}%
\end{thebibliography}%
\end{document}